\documentclass[lettersize,journal]{IEEEtran}
\usepackage{amsmath,amsfonts,amsthm}
\usepackage{algorithmic}
\usepackage{algorithm}
\usepackage{array}
\usepackage{subcaption}
\usepackage{textcomp}
\usepackage{stfloats}
\usepackage{url}
\usepackage{verbatim}
\usepackage{graphicx}
\usepackage{cite}
\usepackage{booktabs}
\usepackage{diagbox}
\usepackage[hidelinks]{hyperref}

\usepackage{multirow}

\usepackage{tcolorbox}
\usepackage{xcolor}
\usepackage{tikz}
\usetikzlibrary{arrows.meta,positioning,calc,decorations.markings,matrix}

\usepackage{aliascnt}
\usepackage[capitalise]{cleveref}
\newtheorem{definition}{Definition}

\newaliascnt{proposition}{definition}
\newtheorem{proposition}[proposition]{Proposition}
\aliascntresetthe{proposition}
\crefname{proposition}{Proposition}{Propositions}
\Crefname{proposition}{Proposition}{Propositions}

\newaliascnt{corollary}{definition}
\newtheorem{corollary}[corollary]{Corollary}
\aliascntresetthe{corollary}
\crefname{corollary}{Corollary}{Corollaries}
\Crefname{corollary}{Corollary}{Corollaries}

\begin{document}

\def\gA{{\mathcal{A}}}
\def\gB{{\mathcal{B}}}
\def\gC{{\mathcal{C}}}
\def\gD{{\mathcal{D}}}
\def\gE{{\mathcal{E}}}
\def\gF{{\mathcal{F}}}
\def\gG{{\mathcal{G}}}
\def\gH{{\mathcal{H}}}
\def\gI{{\mathcal{I}}}
\def\gJ{{\mathcal{J}}}
\def\gK{{\mathcal{K}}}
\def\gL{{\mathcal{L}}}
\def\gM{{\mathcal{M}}}
\def\gN{{\mathcal{N}}}
\def\gO{{\mathcal{O}}}
\def\gP{{\mathcal{P}}}
\def\gQ{{\mathcal{Q}}}
\def\gR{{\mathcal{R}}}
\def\gS{{\mathcal{S}}}
\def\gT{{\mathcal{T}}}
\def\gU{{\mathcal{U}}}
\def\gV{{\mathcal{V}}}
\def\gW{{\mathcal{W}}}
\def\gX{{\mathcal{X}}}
\def\gY{{\mathcal{Y}}}
\def\gZ{{\mathcal{Z}}}

\def\sA{{\mathbb{A}}}
\def\sB{{\mathbb{B}}}
\def\sC{{\mathbb{C}}}
\def\sD{{\mathbb{D}}}
\def\sE{{\mathbb{E}}}
\def\sF{{\mathbb{F}}}
\def\sG{{\mathbb{G}}}
\def\sH{{\mathbb{H}}}
\def\sI{{\mathbb{I}}}
\def\sJ{{\mathbb{J}}}
\def\sK{{\mathbb{K}}}
\def\sL{{\mathbb{L}}}
\def\sM{{\mathbb{M}}}
\def\sN{{\mathbb{N}}}
\def\sO{{\mathbb{O}}}
\def\sP{{\mathbb{P}}}
\def\sQ{{\mathbb{Q}}}
\def\sR{{\mathbb{R}}}
\def\sS{{\mathbb{S}}}
\def\sT{{\mathbb{T}}}
\def\sU{{\mathbb{U}}}
\def\sV{{\mathbb{V}}}
\def\sW{{\mathbb{W}}}
\def\sX{{\mathbb{X}}}
\def\sY{{\mathbb{Y}}}
\def\sZ{{\mathbb{Z}}}

\newcommand{\Range}{\mathrm{Range}}
\newcommand{\Proc}{\mathrm{Proc}}
\newcommand{\x}{\bm{x}}
\newcommand{\z}{\bm{z}}
\newcommand{\g}{\bm{g}}
\newcommand{\m}{\bm{m}}
\newcommand{\vb}{\bm{v}}
\newcommand{\thetab}{\bm{\theta}}
\newcommand{\phib}{\bm{\phi}}
\newcommand{\Cat}{\mathrm{Cat}}
\newcommand{\pib}{\bm{\pi}}
\newcommand{\pub}{\mathrm{pub}}
\newcommand{\priv}{\mathrm{priv}}
\newcommand{\syn}{\mathrm{syn}}
\newcommand{\mub}{\bm{\mu}}
\newcommand{\sigmab}{\bm{\sigma}}
\newcommand{\diag}{\mathrm{diag}}
\newcommand{\KL}{\mathrm{KL}}
\newcommand{\Mone}{\mathrm{M1}}
\newcommand{\AJS}{\mathrm{AJS}}
\newcommand{\JS}{\mathrm{JS}}
\newcommand{\recon}{\mathrm{recon}}
\newcommand{\gen}{\mathrm{gen}}

\newcommand{\dtarget}{D_{\mathrm{target}}}
\newcommand{\mia}{\textsc{MIA}}
\newcommand{\Pin}{P_\mathrm{in}}
\newcommand{\Pout}{P_\mathrm{out}}
\newcommand{\tpr}{\mathrm{TPR}}
\newcommand{\fpr}{\mathrm{FPR}}
\newcommand{\fnr}{\mathrm{FNR}}
\newcommand{\muin}{\mu^{(\mathrm{in})}}
\newcommand{\muout}{\mu^{(\mathrm{out})}}
\newcommand{\sigin}{\sigma^{(\mathrm{in})}}
\newcommand{\sigout}{\sigma^{(\mathrm{out})}}
\newcommand{\adv}{\mathrm{Adv}}

\title{Training-Free Private Synthesis with Validation:\\A New Frontier for Practical Educational Data Sharing}

\author{Hibiki Ito, Chia-Yu Hsu, Hiroaki Ogata\thanks{H. Ito, C-Y, Hsu and H. Ogata are with Kyoto University, Japan.}}

\markboth{Under review.}%
{Shell \MakeLowercase{\textit{et al.}}: A Sample Article Using IEEEtran.cls for IEEE Journals}


\maketitle

\begin{abstract}
While secondary use of growing real-world data (RWD) in education entails the potential for substantial research opportunities, data sharing is mostly restricted by privacy constraints. Differentially private synthetic data generation (DP-SDG) has gained growing attention as a technical solution. However, since educational RWD is fragmented across disparate platforms and institutions in different formats, the implementation of DP-SDG must be tailored to each dataset, requiring substantial engineering efforts. In addition, such RWD is typically small in sample size and high-dimensional, for which deep learning (DL)-based methods are common and require particular expertise for implementation. The small-sample, high-dimensional setting also makes it infeasible to obtain practically acceptable utility in downstream analyses. As a result, despite theoretical advantages, DP-SDG remains far from a practical solution in the education domain. To address this issue, we propose a more practical two-stage method: (1) training-free, LLM-based DP-SDG is performed for sharing synthetic data and (2) \emph{real-data validation}---researchers submit code to remotely validate results---is conducted on demand. Our simple method is designed to be usable for individual data custodians without extensive expertise in DP-SDG. Moreover, the proposed method can also be adapted to multi-shot synthesis settings where data of different learner cohorts are synthesised on a regular basis. We experimentally evaluate this method in both the one-shot and multi-shot synthesis settings using RWD collected over three years. A case study is also conducted with real researchers. The results demonstrate that the LLM-based DP-SDG shows comparable performance to the DL-based baseline despite significantly reduced engineering costs, and that non-DP validation exhibits measurable but moderate privacy leakage. Nonetheless, in the case study researchers reported that on average only 36\% of synthetic findings are validated on real data. Overall, the paper provides a practical and usable method for sharing RWD in education, while highlighting challenges in risk mitigation and epistemic precision. We discuss implications for future research and practitioners.
\end{abstract}


\section{Introduction}\label{sec:intro}

\IEEEPARstart{O}{ver} the past few decades, digital technologies have been widely adopted in education, leading to the accumulation of real-world data (RWD) \cite{Mahajan2015RWD} on learning, learning environments, and educational practices. By complementing experimental data, RWD provides significant opportunities for advancing research to understand learning and empower learners and practitioners \cite{Pan2024SRLA,Okumura2026RWE,Kuromiya2023RWD}. However, due to the privacy-sensitive nature of such RWD, access is typically restricted to trusted researchers or even not provided to third parties \cite{Fischer2020bigdata}. This, in turn, hampers research opportunities and undermines open science practices, thereby constraining the field's impact \cite{Baker2024open,Dawson2019impact}.

Balancing data sharing for the public interest with robust privacy protection remains a substantial practical challenge. Traditional privacy protection techniques such as $k$-anonymity \cite{Samarati1998kanonymity} and $l$-diversity \cite{Machanavajjhala2006ldiversity} cannot robustly control privacy leakage, carrying the risk of unexpected privacy violations including re-identification \cite{Cohen2022attacks}. In particular, micro-level trace data generated by educational technologies often exhibits high privacy risks \cite{Ito2025detailed}, and such risks influence stakeholders' data sharing decisions, thereby affecting downstream secondary use---known as the third-party access effect (3PAE) \cite{Ito20263PAE}.

Differential privacy (DP) \cite{Dwork2006DP} addresses this issue of traditional anonymisation by offering a provable guarantee for privacy leakage. In particular, differentially private synthetic data generation (DP-SDG) has gained growing attention as a technical solution to sharing sensitive data in a privacy-preserving manner \cite{Gadotti2024review,Lu2019SDG}. SDG typically aims to learn distributional properties of a dataset and generate artificial data that resemble the original dataset. While SDG without DP remains vulnerable to privacy risks such as attribute inference attacks \cite{Annamalai2024attribute}, DP-SDG offers a formal privacy guarantee for each individual.

However, despite its theoretical promise, DP-SDG remains far from a practical solution in the education domain. Since most educational contexts generate small, locally collected RWD that is fragmented across disparate platforms in different formats \cite{Nguyen2024small,Medeiros2025datagov}, such RWD is typically small in sample size. Additionally, analysing learning processes within students requires high-dimensional data such as time series and trace data, rather than cross-sectional tabular data \cite{Saqr2026idio}, requiring deep learning (DL)-based methods for private synthesis to capture temporal patterns \cite{Mao2024DPtime}. Nonetheless, implementing DL-based DP-SDG tailored to each dataset requires substantial engineering effort, which would not generally be manageable for individual data custodians---including K-12 schools, higher education institutions, educational technology companies---or researchers without expertise in DP and DL. Moreover, obtaining practically acceptable utility from generated data becomes infeasible for small and high-dimensional data \cite{MontoyaPerez2024FP,Ullman2011hardness}. With limited fidelity of synthetic data to real data, analytical conclusions derived from synthetic data may easily lead to false discoveries \cite{vanDerLinden2025synthhealth,MontoyaPerez2024FP}.

To address this issue, we propose a more practical two-stage method: 1) training-free, LLM-based DP-SDG is performed for sharing synthetic data and 2) non-DP real-data validation---researchers remotely run validation code and obtain output---is conducted on demand. With this two-stage method, even individual researchers and data custodians may be able to share small, high-dimensional sensitive data with privacy protection while downstream researchers can ensure the validity of their findings through real-data validation.

Moreover, our method can be adapted to multi-shot synthesis settings, an education-specific synthesis setting where similar datasets generated by different learner cohorts are shared on a regular basis (e.g. annually) \cite{Ito2026CAPS}. It is important to adapt private synthesis to a multi-shot setting because simply repeating the same one-shot synthesis is suboptimal. Additionally, multi-shot synthesis supports design-based research (DBR) in learning analytics (LA) by iteratively sharing data \cite{Ito2026CAPS}.

Taken together, the two-stage method offers a highly practical and usable method for practitioners to share sensitive data that would otherwise remain in enclaves, potentially expanding research opportunities such as DBR and promoting open science in the education domain. We experimentally evaluate the proposed method using RWD collected over three years in a lower-secondary school. We also conduct a case study of secondary use with four researchers for studying the real-world utility of the proposed method, which has recently been recognised as a valuable yet scarce approach in the DP community \cite{Rosenblatt2023parity}. Overall, the contribution of this paper is twofold:
\begin{itemize}
    \item We propose a practical two-stage method for sharing RWD in education that can be used in a one-shot synthesis setting and adapted to multi-shot settings.
    \item We evaluate the two-stage method through authentic RWD and a case study of secondary use with real researchers.
\end{itemize}

\section{Related work}\label{sec:related-work}

\subsection{Data sharing in education}\label{sec:data-sharing-education}

Educational technology research, including LA and educational data mining (EDM), has been inherently constrained by privacy issues \cite{Pardo2014ethical}.
Thus, there have been ongoing yet limited efforts for technologically mitigating these issues \cite{Liu2023review,Joksimovic2022privacy}. Among them, pseudonymisation---removing personally identifiable information like names and emails---is the most common approach to data sharing in practice \cite{Hutt2023forgotten}, including LA infrastructures \cite{Flanagana2018infra,Wijerathne2024EREDA} and a MOOC plugin \cite{Torre2020ELAT}. However, pseudonymised educational data remain highly vulnerable to re-identification \cite{Yacobson2021deident,Vatsalan2022risk,Ito2025detailed}. Thus, more advanced anonymisation methods such as $k$-anonymity have also been investigated \cite{Angiuli2016kanon,Kyritsi2019privacy,Stinar2024kanon}. For example, the Open University Learning Analytics Dataset (OULAD) \cite{Kuzilek2017OULAD} and the HarvardX dataset \cite{HarvardX2014data} were released using $k$-anonymity. In addition, multiple tools have been developed to share educational data based on quantified risks, such as the ARX tool \cite{Prasser2020ARX} and the Re-identifier Risk Ready Reckoner (R4) \cite{CSIRO2019R4}.

Nonetheless, despite the prevalence in practice, these traditional methods fail to robustly control privacy risks, being unpredictably vulnerable to inferential privacy attacks such as attribute inference \cite{Fredrikson2014AI} and membership inference attacks (MIA) \cite{Shokri2017MIA}. Although SDG has emerged as a modern solution and has also been applied to educational data \cite{Flanagan2022synth,Khalil2025CTGAN,Vie2022synth,Iloh2025CTGAN,Bautista2021GAN}, it still remains susceptible to unexpected privacy risks without formal guarantees such as DP. Recently, DP-SDG has been applied to educational data \cite{Liu2025fairprivate,Liu2025DPGAN,Zhan2024DPGAN,Kesgin2025FairSYN,Ito2026CAPS}. A concurrent study proposed the cyclic adaptive private synthesis (CAPS) that adapts variational autoencoder (VAE)-based DP-SDG to multi-shot settings \cite{Ito2026CAPS}. While most work assumes large-scale low-dimensional data is available, CAPS delves into more practical educational contexts with small and high-dimensional data \cite{Ito2026CAPS}. Nonetheless, the VAE-based framework is not flexible to changes in data schema and also requires substantial engineering efforts in model development. In this paper, building on this line of research, we aim to develop a more practical method such that individual data custodians and researchers can conduct data sharing while still maintaining more robust privacy protection than traditional anonymisation.

\subsection{Training-free DP-SDG}\label{sec:bg-stage1}

Private synthesis of high-dimensional structured data such as time series and event sequences predominantly involves training neural networks \cite{Mao2024DPtime}. Although there are some statistics-based methods that privately synthesise longitudinal data without DL-based generative models \cite{He2024onlineSDG,Bun2024continual,Gu2025continuous,Wang2025DPTrajPM}, they focus on event-level DP, which we do not adopt in this paper because it can pose uneven privacy guarantees for different individuals and thus is arguably not ethical. Instead, we focus on user-level DP, where the privacy unit is an individual and the same formal privacy is guaranteed universally across individuals.

Some training-free methods with user-level DP are technically applicable to high-dimensional structured data. Private evolution (PE) is a powerful method to generate synthetic images \cite{Lin2023PEimage,Wang2025PEimagetext,Zhang2025PCEvolve,Lin2025PEsim} and text \cite{Xie2024PEtext}. The Private Signed Measure Mechanism (PSMM) is a related method in theory \cite{He2023PSMM}. However, the effectiveness of these methods is limited for structured data \cite{Swanberg2025PEtabular}, and they can also fail to improve or even degrade generated data for small, high-dimensional regime \cite{Gonzalez2025PEconverge}.

\subsection{Real-data validation}\label{sec:bg-stage2}

The remote data analysis, where a researcher submits code to get output without seeing real data, is commonly adopted in trusted research environments (TRE) \cite{Weise2024TRE}. For example, OpenSAFELY, a secondary-use platform for electronic health records in the UK, provides random mock data to researchers to develop analysis and allow for remote execution with output checking \cite{Nab2024OpenSAFELY}. This approach would be suitable for domains like medicine, where researchers have prior knowledge or beliefs about underlying distributions, because most analysis is confirmatory based on researchers' hypotheses. However, educational data is contextualised to local practice and thus typically requires exploratory analyses for hypothesis generation, which is not possible with random data. Validating every single exploratory analysis with output checking poses infeasible operational costs and unnecessary privacy loss.

The US Census Bureau conducted pilot studies of sharing synthetic data while allowing for real-data validation \cite{Drechsler2014synthLB,USCensusBureau2018SIPP}. However, these studies do not use DP for generating synthetic data. As mentioned previously, non-DP synthetic data can be unpredictably vulnerable to privacy attacks \cite{Stadler2022groundhog}.

To the best of our knowledge, our work is the first to combine DP-SDG and non-DP validation for data sharing. This is perhaps because non-DP outputs break the DP guarantee and the privacy risk of the entire data sharing process cannot be analysed in the DP framework. Nonetheless, as we discuss in \cref{sec:two-stage}, focusing on the privacy framework of $f_{n,\sD}$-\emph{worst-case membership inference privacy} ($f_{n,\sD}$-WMIP) \cite{Ito2024thesis}, we show that the entire process of the two-stage method can still account for the privacy loss for each individual.

\section{Two-stage data sharing}\label{sec:two-stage}

We describe the proposed two-stage method and how it is adapted to a multi-shot setting (\cref{fig:proposed-method}).

\subsection{Preliminaries}

To formulate our two-stage method, we first define approximate DP and $f$-DP. We say that datasets $D$ and $D^\prime$ are adjacent if they differ by at most one individual. Throughout, we assume that adjacent datasets have the same number of individuals. The standard approximate DP is defined below.

\begin{definition}[Differential privacy \cite{Dwork2006approxDP}]\label{def:dp}
    An algorithm $\gA$ is $(\varepsilon, \delta)$-differentially private if for all $\gS \subseteq \Range(\gA)$ and for all adjacent datasets $D$ and $D^\prime$:
    \begin{equation}
        \Pr(\gA(D) \in \gS) \leq e^\varepsilon \Pr(\gA(D^\prime) \in \gS) + \delta,
    \end{equation}
    where probabilities are over the randomness in the algorithm $\gA$.
\end{definition}

Next, to define $f$-DP, a variant of DP that generalises the approximate DP, we define trade-off functions based on the false negative rate (FNR) and false positive rate (FPR) of a distinguishability problem between two distributions.

\begin{definition}[Trade-off functions \cite{Dong2022GDP}]
    For probability distributions $P_0$ and $P_1$ in the same space, consider the hypothesis test problem $H_0:P_0\ vs.\ H_1:P_1$. Then the trade-off function $T(P_0, P_1):[0,1] \rightarrow [0,1]$ is defined as
    \begin{equation}
        T(P_0, P_1)(\alpha) = \inf\{\fnr \mid \fpr = \alpha\},
    \end{equation}
    where infimum is taken over all rejection rules.
\end{definition}

Let $\gA(D)$ denote the output distribution of algorithm $\gA$ on dataset $D$. $f$-DP is defined by the trade-off function of distinguishability between $\gA(D)$ and $\gA(D^\prime)$ where $D$ and $D^\prime$ are adjacent.

\begin{definition}[$f$-differential privacy \cite{Dong2022GDP}]\label{def:f-dp}
    Let $f$ be a trade-off function. An algorithm $\gA$ is $f$-differentially private ($f$-DP) if for all adjacent datasets $D$ and $D^\prime$:
    \begin{equation}
        T(\gA(D), \gA(D^\prime)) \geq f.
    \end{equation}
\end{definition}

Here we denote $f_0 \geq f_1$ for functions $f_0$ and $f_1$ defined on $[0, 1]$ if $f_0(x) \geq f_1(x)$ for all $0 \leq x \leq 1$. In particular, let $G_\mu$ be a trade-off function defined as
\begin{equation}
    G_\mu=T(\gN(0,1),\gN(\mu,1)).
\end{equation}
Then Gaussian differential privacy (GDP) is defined as follows.
\begin{definition}[$\mu$-Gaussian differential privacy \cite{Dong2022GDP}]\label{def:gdp}
    An algorithm $\gA$ is $\mu$-Gaussian differentially private ($\mu$-GDP) if it is $G_\mu$-DP, that is, if for all adjacent datasets $D$ and $D^\prime$:
    \begin{equation}
        T(\gA(D), \gA(D^\prime)) \geq G_\mu.
    \end{equation}
\end{definition}

That is, if algorithm $\gA$ satisfies $\mu$-GDP, distinguishing $\gA(D)$ and $\gA(D^\prime)$ is harder than distinguishing $\gN(0,1)$ and $\gN(\mu,1)$ for any adjacent datasets $D$ and $D^\prime$.

Since the DP guarantee does not hold for the two-stage method due to non-DP validation, we instead focus on $f_{n,\sD}$-WMIP. To formulate it, we first define the per-example membership inference (MI) game in the context of real-data validation.

\begin{definition}[Per-example MI game adopted from \cite{Ye2022MIA}]\label{def:mi-game}
    The per-example membership inference game between a data custodian and an adversary proceeds as follows:
    \begin{enumerate}
        \item The adversary submits code (algorithm) $\gA$ to the data custodian for real-data validation.
        \item The custodian samples a dataset $\dtarget \sim \sD^{n-1}$ and is given a fixed target example $x$ (pre-sampled from $\sD$).
        \item The custodian flips an unbiased bit $b\leftarrow\{0,1\}$. If $b=1$ (IN), then the custodian runs the code to get output $Y=\gA(\dtarget \cup \{x\})$. If $b=0$ (OUT), then the custodian runs the code to get output $Y=\gA(\dtarget \cup \{x^\prime\})$ for some $x^\prime \sim \sD$ such that $x \not= x^\prime$. Then the custodian sends $x$ and $Y$ to the adversary.
        \item Given access to $\sD$, the adversary outputs $\mia(x, Y, \sD) \in \{0, 1\}$. If $\mia(x, Y, \sD)=b$, then the adversary wins the game; otherwise, the custodian wins.
    \end{enumerate}
\end{definition}

We evaluate the adversary's success in the above game in terms of the trade-off functions of the following hypothesis test:
\begin{equation}\label{eq:mia-hypothesis-test}
    H_0: x \notin \dtarget \quad vs. \quad H_1: x \in \dtarget.
\end{equation}
That is, the membership of a target individual is identified by rejecting the null hypothesis. Now let us denote
\begin{align}
    \gA(D) \sim \Pin(x, \gA, \sD, n) \quad &\text{if} \ D \sim \sD^n \ \text{with}\ x \in D, \\
    \gA(D) \sim \Pout(x, \gA, \sD, n) \quad &\text{if} \ D \sim \sD^n \ \text{with}\ x \notin D.
\end{align}

\begin{definition}[$f_{n,\sD}$-worst-case membership inference privacy \cite{Ito2024thesis}]\label{def:f-wmip}
    Let $f_{n,\sD}$ be a trade-off function. An algorithm $\gA$ is $f_{n,\sD}$-worst-case membership inference private ($f_{n,\sD}$-WMIP) with respect to the dataset size $n$ and the underlying distribution $\sD$ if for all $x \sim \sD$:
    \begin{equation}
        T(\Pout(x, \gA, \sD, n), \Pin(x, \gA, \sD, n)) \geq f_{n,\sD}.
    \end{equation}
\end{definition}

Similar to GDP, we define Gaussian worst-case membership inference privacy (GWMIP).
\begin{definition}[Gaussian worst-case membership inference privacy \cite{Ito2024thesis}]\label{def:GWMIP}
    An algorithm $\gA$ is $\mu_\sD(n)$-Gaussian worst-case membership inference private ($\mu_\sD(n)$-GWMIP) with respect to the dataset size $n$ and the underlying distribution $\sD$ if it is $G_{\mu_\sD(n)}$-WMIP, that is, if for all $(x, y_x) \sim \sD$:
    \begin{equation}
        T(\Pout(x, \gA, \sD, n), \Pin(x, \gA, \sD, n)) \geq G_{\mu_\sD(n)}.
    \end{equation}
\end{definition}

Note that, while $f$-DP focuses on distinguishability of adjacent datasets, $f_{n,\sD}$-WMIP addresses distinguishability between IN and OUT datasets, weakening the threat model to a more realistic one. Additionally, while DP is a characteristic of an algorithm, WMIP depends not only on an algorithm but also the dataset size $n$ and the underlying distribution $\sD$.

Finally, the following result ensures the existence of the worst-case individual.
\begin{proposition}[Existence of the worst-case example\cite{Ito2024thesis}]\label{prop:closed-bounded}
    Let $\gA$ be an algorithm with domain $\gD$, $\sD$ be a distribution over $\gD$, and $n \in \sN$ be an arbitrary dataset size. Assume $\gD$ is closed and bounded. Then for all $\alpha\in[0,1]$ there exists $(x^*, y_{x^*})\in\gD$ such that for all $(x, y_{x}) \in \gD$
    \begin{multline}
        T(\Pout(x, \gA, \sD, n), \Pin(x, \gA, \sD, n))(\alpha) \\
        \geq T(\Pout(x^*, \gA, \sD, n), \Pin(x^*, \gA, \sD, n))(\alpha).
    \end{multline}
\end{proposition}

\begin{figure*}[tb]
  \center
  \includegraphics[width=0.8\textwidth]{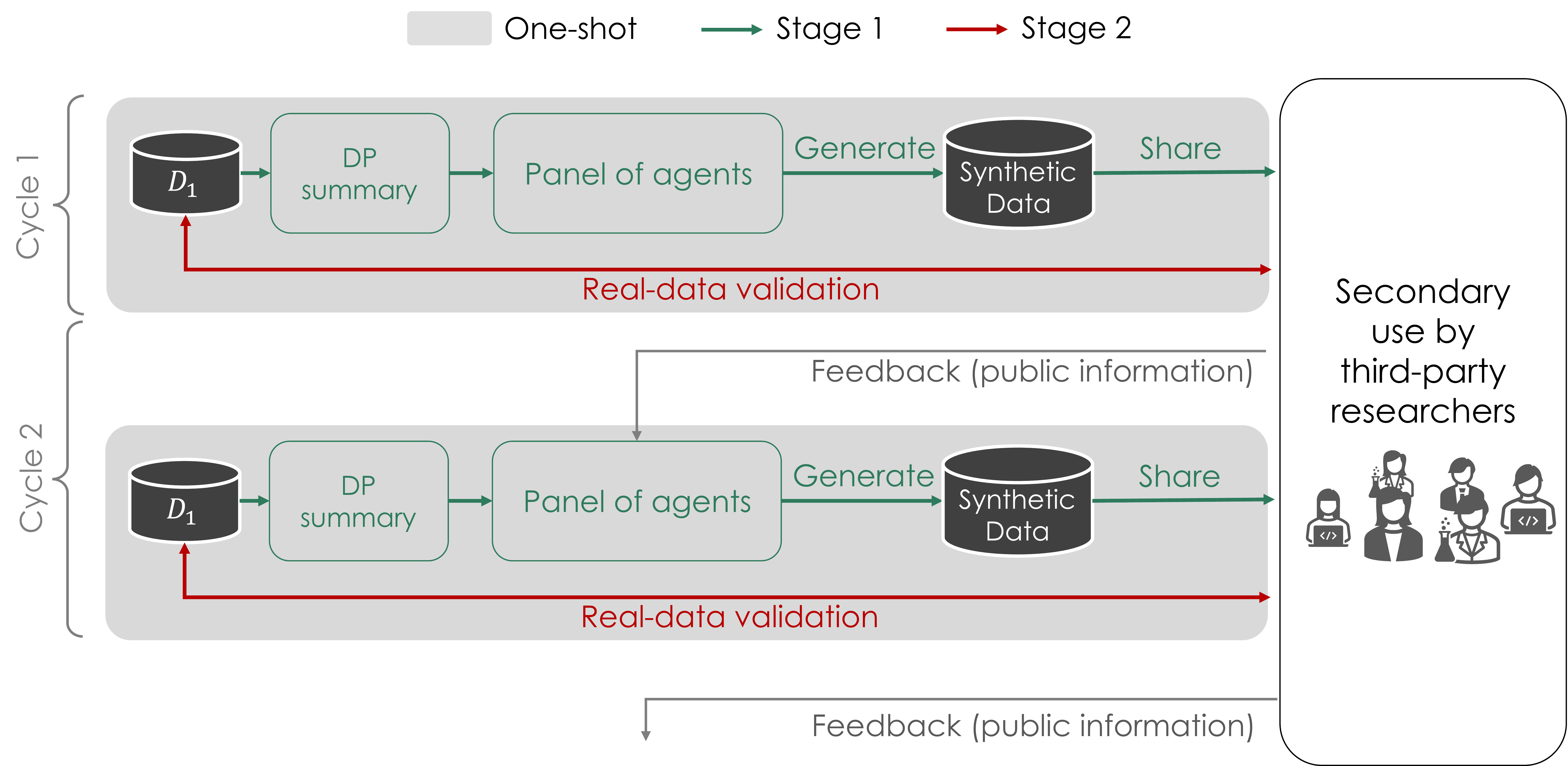}
  \caption{The proposed two-stage method in a multi-shot setting. $D_1,D_2,\dots$ are private datasets of different learner cohorts. The initial cycle corresponds to a standard one-shot setting. From the second cycle, researchers give feedback on how findings on synthetic and real data differed in the prior cycle.}
  \label{fig:proposed-method}
\end{figure*}

\subsection{Stage 1: LLM-based DP-SDG}

A data custodian starts by generating DP summary statistics, like counts and means, of data to be shared. This can be tailored to downstream analyses by third-party researchers. Open-source tools are available for basic DP queries, such as OpenDP \cite{Gaboardi2020OpenDP}. Then, using the DP statistics, LLMs generate synthetic data, which is treated as post-processing of the DP release. As LLM-driven synthetic data generation is an emerging topic \cite{Long2024LLMsynth}, we expect more diverse techniques suitable for various educational RWD are developed in the future. As a baseline method, we give multiple LLMs, a panel of agents, contexts of data collection and the schema and tell them to independently generate a Python script that produces synthetic data using the DP summary. This is a similar approach in public data generation in prior work \cite{Hod2025surrogate,Ito2026CAPS} but with private information. Then we aggregate generated data and randomly sample a desired number of points. Note that Stage 1 alone satisfies DP.

\subsection{Stage 2: Real-data validation}

Third-party researchers who have received the generated data develop analysis. If researchers need real-data validation in order to e.g. publish a paper, they send a request to the data custodian to run their code on real data. Open source tools like PySyft\footnote{\url{https://github.com/OpenMined/PySyft}} are available for this purpose. Output must not include individual raw data and must be safe aggregates such that an adversarial third party cannot trivially reconstruct real data (e.g. by differencing multiple output values). Thus, statistical disclosure control (SDC) is applied, depending on the data sharing policy \cite{Griffiths2024SDC}. There are also open-source tools that help practitioners with SDC, such as statbarn \cite{Green2024statbarn} and QuerySnout \cite{Cretu2022QuerySnout}. We do not apply DP in this stage to ensure validity of findings.

\subsection{Privacy loss}\label{sec:theory}

Though the two-stage data sharing does not provide a DP guarantee, we may analyse privacy loss in terms of WMIP. Let us denote the tensor product of two trade-off functions $f=T(P_0,P_1)$ and $g=T(Q_0,Q_1)$ by
\begin{equation}
    f \otimes g = T(P_0\times Q_0, P_1\times Q_1).
\end{equation}
Then we have the following.

\begin{proposition}
  If Stage 1 satisfies $f$-DP and Stage 2 satisfies $g_{n,\sD}$-WMIP, then the two-stage data sharing satisfies $f\otimes g_{n,\sD}$-WMIP.
\end{proposition}
\begin{proof}
  By Theorem 5.2 of \cite{Ito2024thesis}, Stage 1 satisfies $f$-WMIP. The conclusion follows from the composition theorem of WMIP (Proposition 5.10 of \cite{Ito2024thesis}).
\end{proof}

In particular, we focus on GWMIP \cite{Ito2024thesis} for reporting privacy loss by a single parameter. The below corollary follows from the composition theorem of GWMIP \cite{Ito2024thesis}.

\begin{corollary}\label{cor:composition}
  If Stage 1 satisfies $\mu$-GDP and Stage 2 satisfies $\nu_\sD(n)$-GWMIP, then the two-stage data sharing satisfies $\sqrt{\mu^2 + \nu_\sD(n)^2}$-GWMIP.
\end{corollary}

Some important remarks are as follows. Unlike DP, WMIP cannot be enforced by a mechanism. Thus, WMIP is necessarily empirical. Nonetheless, we still focus on WMIP because we can analyse the privacy loss of the entire two-stage data sharing, not each stage separately.

If the two-stage data sharing satisfies $\mu$-GWMIP with small $\mu$, the risk of membership inference attacks (MIA) defined in the MI game (\cref{def:mi-game}) becomes limited. That is, each individual's participation in data sharing is protected from being inferred by third parties. Success of MIA can be significant privacy violation when the learner group has a sensitive attribute such as having learning disabilities. Though the MI game has a strong assumption that the adversary has access to the target data $x$, protection against MIA with this strong assumption implies protection against attribute inference with weaker assumptions \cite{Yeom2018privacyrisk}. In an attribute inference attack, an adversary attempts to infer a sensitive attribute (e.g. exam scores or certain learning behaviour) using auxiliary information.

However, WMIP is separable from resilience against reconstruction attacks \cite{Salem2023privacygames}. This is a limitation of the two-stage method involving non-DP real-data validation and what data custodians must take into account. While success of a reconstruction attack does not imply re-identification, part of real data can be reconstructed, potentially violating a data sharing policy. In other words, while individual privacy remains protected, dataset-level properties can possibly be disclosed to third parties unwittingly.

\subsection{Adapting to multi-shot settings}

When the two-stage method is adapted to multi-shot setting, an important additional step is that researchers give the data custodian feedback on how synthetic findings differ from the output of real-data validation. This information is considered public because both synthetic data and validation output may be published. Then the feedback is provided to the panel of LLM agents to improve synthetic data generation in the next cycle. By the post-processing property of $f$-WMIP \cite{Ito2024thesis}, the privacy loss of each cycle remains the same when adapted to multi-shot settings.

In our implementation, an LLM first summarises submitted code, a corresponding research question and feedback on how results on synthetic and real data differ into concise natural language. Then, another LLM converts the summary into suggestions for improving the subsequent generation cycle, such as ``synthetic data should better match real-world heavy-tail behaviour.''

\section{Evaluation method}\label{sec:method}

While the proposed method significantly reduces engineering costs, the utility of synthetic data generated in the Stage 1 remains unclear. In addition, the Stage 2 introduces additional privacy loss that is not guaranteed to be minimal. Thus, we set the following guiding research questions for evaluation\footnote{Code and prompts used for evaluation are available at \url{https://anonymous.4open.science/r/anonymous-2stage-sharing}.}:
\begin{description}
  \item[RQ1] What is the utility of synthetic data generated by the proposed LLM-based method (Stage 1) compared to DL-based one in the one-shot and the multi-shot settings?
  \item[RQ2] What is the empirical privacy risk of the entire two-stage data sharing?
\end{description}

\subsection{Data and contexts}\label{sec:data}

We focus on learning habits research (a.k.a habits mining \cite{Hsu2023habitsmining}) as a representative topic of secondary use of RWD in education. Empirical evidence suggests that the formation of learning habits is significantly correlated with learning outcomes such as academic achievement \cite{ShirvaniBoroujeni2019MOOChabits,Magulod2019habits,Ricker2020habits} and productivity \cite{Hsu2024productivity}. Furthermore, learning habits like daily routines are privacy-sensitive, and individual trajectories can be re-identified by inspecting longitudinal patterns unless not carefully protected \cite{Vatsalan2022risk,Ito2025detailed}. Recent work also suggests that learning habits analyses could be sensitive to opt-outs incurred by privacy concerns \cite{Ito20263PAE}. Thus, protecting individual privacy about learning habits while enabling secondary use for understanding and improving learning for all is of paramount importance.

We use RWD collected through the BookRoll system \cite{Ogata2017BR}, an e-book reader deployed at a lower-secondary school, over three years from 2022 to 2024. The system is used for mathematics studying in and out of classes. Learning materials and exercises are distributed to students as PDFs through BookRoll, and students interact with them. BookRoll collects interaction logs---such as opening/closing materials, highlighting text, adding handwritten memo---in xAPI format and stores them in a learning record store (LRS).

For each year we obtained data of the first semester (17 weeks) of 120 7th-grade students and their end-of-semester exam scores of mathematics. Then, we estimated sessions using the method of \cite{Hsu2023chronotypes} and aggregated time in minutes spent on BookRoll for four time windows \cite{Ricker2020habits}: morning (05:00-11:59), afternoon (12:00-16:59), evening (17:00-23:59) and overnight (00:00-04:59). The exam scores are evenly binned into three classes: low, average and high. That is, for each student we have $17 \text{ weeks} \times 4 \text{ time windows} = 68$ features and one outcome label, which is fairly high-dimensional data considering the sample size $n=120$.

In the case study, we share three synthetic datasets generated by the proposed LLM-based method with four researchers in the chronological order, emulating a three-year secondary-use project. That is, we first generate and share 2022 synthetic data (Stage 1), and researchers locally analyse and submit real-data validation requests (Stage 2). Then, we generate 2023 synthetic data based on the approved requests for 2022 data and share with researchers (Stage 1), and researchers conduct real-data validation similarly (Stage 2). We repeat this for 2024 data. Informed consents for research-purpose data collection are obtained from both the secondary-school students and researcher participants.

\subsection{One-shot synthesis}

\subsubsection{DL baseline}

For the baseline, we use the semi-supervised generative model of Kingma et al.\ \cite{Kingma2014SSL}, which is the one-shot part of the CAPS framework. First, we preprocess the data by normalising each time window to $[0,1]$ and then applying $\log(1+x)$ for stability. Then, a larger VAE, called M1, is pretrained with public data, and a smaller conditional VAE, called M2, is trained with private data on top of M1. See Appendix \ref{sec:model-architecture} for details of the model architecture and training.

For pretraining M1, we use LLM-generated data as surrogate public data \cite{Hod2025surrogate}. We use \verb|gpt-5.1|, \verb|gemini-3.0-pro| and \verb|claude-opus-4.5| for this task. A prompt shown in \cref{fig:llm-prompt} is given to each model, generating 20,000 data points for each academic achievement class (180,000 in total), and 100,000 points are sampled from the pool of all generated points. We use 1D convolutional layers for the encoder and decoder of M1 with a Gaussian latent variable $z_1$. Since the data is zero-inflated, the decoder predicts (i) a Bernoulli gate for exact zeros and (ii) the parameters of a Gamma distribution for positive values at each feature-time location.

For training M2, semi-private semi-supervised learning (SPSSL) with the Adam optimiser \cite{Ito2026CAPS} is used. Synthetic data is generated from this stacked M1+M2 model for sharing with third-party researchers. M2 consists of a linear classifier that predicts the label $y$ of a latent $z_1$ encoded by M1 and a conditional VAE that learns $p(z_1 \mid z_2, y)$ with a Gaussian latent $z_2$. Using the Opacus library \cite{Yousefpour2021Opacus}, we implement DP-Adam \cite{Abadi2016DPSGD,Kingma2015Adam} by accounting privacy with R\`enyi DP (RDP) \cite{Mironov2017RDP}. Throughout, We set $\varepsilon=1.0$ and $\delta=10^{-3}$. 

For both M1 and M2, hyperparameter optimisation (HPO) is performed by the Tree-structured Parzen Estimator (TPE) sampler \cite{Bergstra2011TPE} with 20 trials using the Optuna library \cite{Akiba2019Optuna}. For M1, LLM-generated data is used for HPO as surrogate public data \cite{Hod2025surrogate}.

\subsubsection{Proposed two-stage method}

In Stage 1, we first calculate the average proportion of zeros for each time window over the 17 weeks to reflect zero-inflation. We add noise to them by the Gaussian mechanism \cite{Dwork2006DP,Balle2018Gaussian} so the output satisfies DP with $\varepsilon=1.0$ and $\delta=10^{-3}$. Then, we give the prompt in \cref{fig:llm-prompt} with the DP summary to the three LLMs. The LLMs independently produce 1000 points, and we randomly choose 120 from the pooled sample for sharing with researchers. Similarly to prior work of PE \cite{Lin2023PEimage}, we generate synthetic data for each label and combine them at the end. 

In Stage 2, researchers submit validation requests containing code, a corresponding research question and the description of output on synthetic data, and we review them before running on real data so that the code does not output individual-level data or any information that allows for trivially inferring individual data.

\subsection{Multi-shot synthesis}

\subsubsection{DL baseline}

After sharing 2022 data with the above one-shot setting (Cycle 1), we update M1 using the generative replay method \cite{Shin2017replay}. We sample 5000 synthetic data points from both M1 and M1+M2 stacked model, and train M1 on this data. After that, we train a fresh M2 on top of the updated M1 using 2023 data. Then synthetic data is generated from this stacked model for sharing (Cycle 2). We repeat the same for year 2024 (Cycle 3).

\subsubsection{Proposed two-stage method}

For the two-stage data sharing, Stage 1 slightly differs between the one-shot and multi-shot settings: we additionally give suggestions to improve synthesis based on researchers' feedback in the previous cycle. Stage 2 remains the same.

\subsection{Utility metrics}\label{sec:method-utility}

As a fidelity metric, we use the average Jensen-Shannon (AJS) divergence between synthetic and real data \cite{Ito2026CAPS,Stenger2024eval}. For each time series $x$ of four time windows, let $h(x)$ be a $4\cdot 6=24$ dimensional vector containing the median, mean, standard deviation, variance, maximum value and root mean square of each time window. Then the AJS divergence between a real dataset $D_t$ and a synthetic dataset $D_t^\syn$ for cycle $t$ is given as
\begin{align}
    \AJS(D_t, D_t^\syn) =& \frac{1}{3} \sum_{c=1}^3 \left( \frac{1}{24} \sum_{h=1}^{24} \JS (\widehat{P}_{t,c}^{(d)}, \widehat{Q}_{t,c}^{(d)} ) \right) \\\label{eq:ajs}
    \widehat{P}_{t,c}^{(d)} =& \left\{h_d(x) \mid (x,y)\in D_t, y=c\right\}, \\
    \widehat{Q}_{t,c}^{(d)} =& \left\{h_d(x) \mid (x,y)\in D_t^{\mathrm{syn}}, y=c\right\}.
\end{align}
where $\JS$ denotes Jensen-Shannon divergence between two empirical distributions, $h_d(x)$ is the $d$-th dimension of the vector $h(x)$ and $c=1,2,3$ are the academic achievement labels.

Additionally, inspired by the \emph{epistemic parity} metric introduced by Rosenblatt et al.\ \cite{Rosenblatt2023parity}, we define \emph{epistemic precision} (EPrec) as subjective evaluation by researchers about how much of synthetic findings are validated on real data. While epistemic parity is defined as an evaluation metric of an SDG algorithm with respect to a specific dataset and a specific finding, we define EPrec as a metric to evaluate an SDG algorithm with respect to a dataset, aggregating across findings. This is because we wish to evaluate how much the algorithm prevents false discoveries on a dataset. In the case study, we shared synthetic data generated by the proposed LLM-based method and asked researchers to report the EPrec of each validation request. We adopted three-scale reporting to reduce sensitivity: $\textrm{EPrec} = 0.0$ indicates that none of the findings are validated on real data, $\textrm{EPrec} = 0.5$ indicates the findings are partially validated, and $\textrm{EPrec} = 1.0$ indicates most of the findings are validated.

\begin{figure*}[t]
    \centering
    \begin{tcolorbox}[
        boxrule=0.8pt,
        arc=2mm,
        left=3mm,right=3mm,top=2mm,bottom=2mm,
        width=0.95\linewidth
    ]
    \footnotesize\ttfamily
    You are a professional synthetic-data generator.\\
    TASK\\
    Write a Python script to generate a realistic dataset of \{\{label\}\}-achieving lower-secondary-school students' learning behaviour during 17 weeks of mathematics study.\\
    DATA-SHAPE\\ 
    - N students\\  
    - 17 weeks (week = 1...17)\\
    - 4 daily attributes: overnight (00:00-04:59), morning (5:00-11:59), afternoon (12:00-16:59), evening (17:00-23:59)\\
    - Minutes spent on studying per week per attribute: integer values 0-300 (overnight), 0-420 (morning), 0-300 (afternoon) and 0-420 (evening)\\
    - Data shape: (N, 4, 17)\\
    STATISTICS\\
    - Include a realistic number of edge-case students\\
    - Assume no holidays, breaks or exam periods\\
    - Ensure a realistic distribution that reflects 12 to 13 year-old {label}-achieving students' behaviour\\
    - Data is derived from interaction logs of an e-book system, where students interact with learning materials and the system was used in classrooms and at home\\
    - Overall zero inflation: the proportions of zero values are shown below.\\
    $$\texttt{\{\{zero props\}\}}$$
    OUTPUT\\
    - A Python script to generate the dataset. The script should generate a numpy array with shape (N, 4, 17) and save it to \verb|<path_to_save_file>| with \verb|<random_seed>| by running \verb|python generate.py --save_path <path_to_save_file> --num_samples N --seed <random_seed>|.\\
    DELIVERY RULES\\
    - Do not print any explanatory text, previews, or stats
    \end{tcolorbox}
    \caption{Prompt template provided to LLMs to generate synthetic data. \{\{label\}\} indicates academic achievement class (high, average, low) and \{\{zero props\}\} are calculated from real data with the Gaussian mechanism. When generating surrogate public data for DL baseline, zero proportions are not given. When the two-stage method is adapted to CAPS, suggestions to improve synthesis from the previous cycle are appended to the prompt.}
    \label{fig:llm-prompt}
\end{figure*}

\subsection{Risk metrics}\label{sec:method-risk}

We evaluate privacy risks of the two-stage data sharing by the MI game (\cref{def:mi-game}). We first manually reformat submitted code $\gA_r$ of $r$-th request so the code outputs a $d_r$-dimensional real vector while containing the same information. Then we resample the real data with replacement to generate 64 IN and 64 OUT \emph{shadow} datasets, containing and not containing the target student $x$. By running the same code $\gA_r$ on the shadow datasets, we obtain empirical distributions of $\Pin(x, \gA_r, \sD, n)$ and $\Pout(x, \gA_r, \sD, n)$. By assuming that these distributions are Gaussians, we perform likelihood ratio tests, similar to the likelihood ratio attack (LiRA) in machine learning \cite{Carlini2022LiRA}. This is a strong attack due to the Neyman-Pearson lemma \cite{Neyman1933NP}, which states the likelihood ratio test is the optimal hypothesis test for a fixed FPR. Assuming that the outputs within and across requests are independent, we have the likelihood ratio for an attacker who has access to $Y_1,\dots,Y_R$ given as
\begin{equation}
  \Lambda_{1:R} := \frac{\prod_{r=1}^R\prod_{i=1}^{d_r} \gN(Y_r^{(i)}; \muin_{r,i}, {\sigin_{r,i}}^2)}{\prod_{r=1}^R\prod_{i=1}^{d_r} \gN(Y_r^{(i)}; \muout_{r,i}, {\sigout_{r,i}}^2)},
\end{equation}
where $Y_r^{(i)}$ is the $i$-th dimension of the output $Y_r = \gA_r(\dtarget)$ and $\gN(\cdot;\mu,\sigma^2)$ denotes a normal pdf. Then, we have
\begin{align}
  \fnr_{1:R}(x) &= \Pr_{\dtarget \sim \sD^n}(\Lambda_{1:R} < \gamma \mid x \in \dtarget) \\
  \fpr_{1:R}(x) &= \Pr_{\dtarget \sim \sD^n}(\Lambda_{1:R} \geq \gamma \mid x \notin \dtarget)
\end{align}
To estimate empirical $\nu_{n,\sD}$-GWMIP for real-data validation, we pick computed FPRs between $2/64$ and $1-2/64$ to avoid unstable edge points, fit $\nu_{n,\sD}$ to each FPR and then report the median as the estimated $\nu_{n,\sD}$. Subsequently, empirical $\nu_{n,\sD}$ is calculated by \cref{cor:composition}. We also report the regret $\delta$ that indicate the fit of $\nu_{n,\sD}$-GWMIP to computed trade-off functions \cite{Kaissis2024regret}.

Additionally, since trade-off functions do not necessarily fit GWMIP, we also report the MIA advantage for the worst-case example $x$ \cite{Yeom2018privacyrisk}:
\begin{multline}
  \adv_{1:R}(x, \gA, n, \sD) \\ 
  = 2 \Pr_{\dtarget \sim \sD^n}(\Lambda_{1:R} \geq \gamma \mid x \in \dtarget) - 1
\end{multline}

\section{Evaluation results}

\subsection{RQ1: Utility of synthetic data in one-shot and multi-shot settings}

We first analyse the utility of the proposed LLM-based DP-SDG in the one-shot setting. \cref{fig:one-shot} shows the AJS for the DL-based baseline method and the proposed method for different datasets. As with other plots, confidence intervals are computed using 1000 bootstrap resamples by running the experiment with five different random seeds. The proposed LLM-based method shows comparable performance across datasets with subtle differences. This indicates that the proposed method operationalises DP-SDG for each educational dataset with comparable utility while significantly reducing the cost of engineering.

\begin{figure}
    \includegraphics[width=0.48\textwidth]{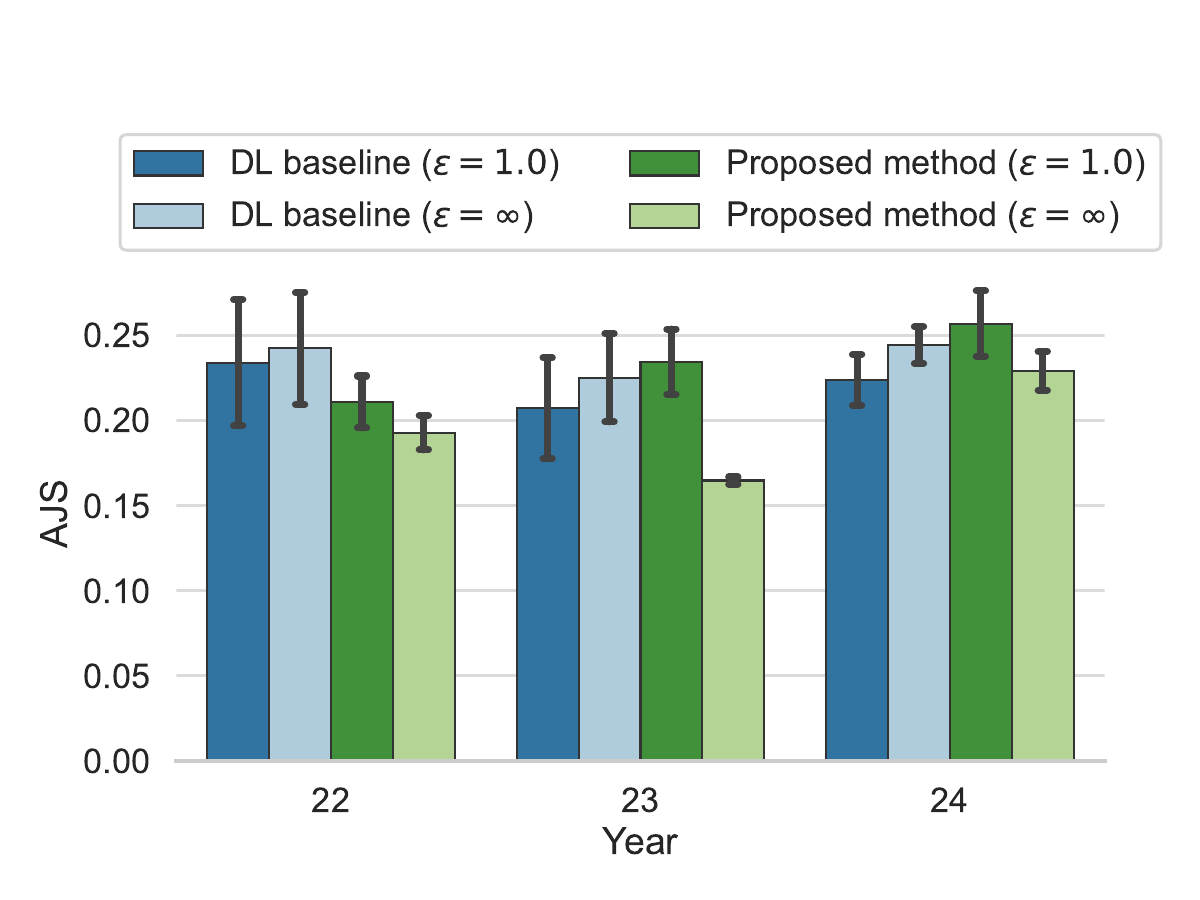}
    \caption{The average Jensen-Shannon divergence (AJS) between real and synthetic data for the one-shot synthesis by the baseline DL-based method and the proposed LLM-based method. Lower AJS indicates better fidelity.}
    \label{fig:one-shot}
\end{figure}

We adapt the two-stage method to the multi-shot setting. \cref{fig:utility-caps} shows the AJS gains (i.e. difference from the one-shot version) over cycles within the multi-shot setting.
We observe similar utility gains between the DL- and LLM-based methods, with both being not significant. This indicates that there is no notable difference in how effectively these methods can be improved by adapting to multi-shot synthesis, implying the need for further investigation of adaption techniques.

\begin{figure}
    \includegraphics[width=0.48\textwidth]{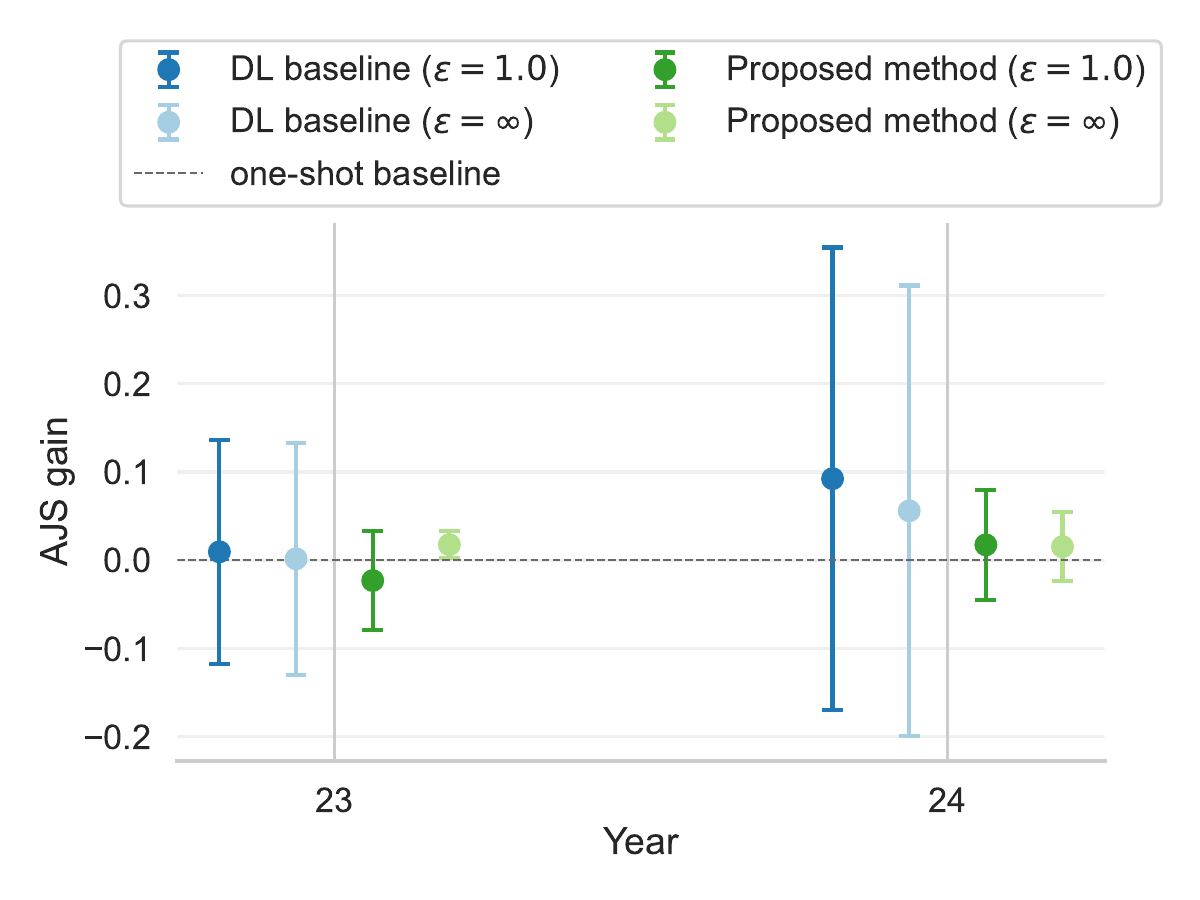}
    \caption{Gains in the average Jensen-Shannon divergence (AJS) by adapting the DL-based baseline method and the proposed two-stage method to the multi-shot synthesis setting. Note that the DP guarantee of generated data for the two-stage method is only with respect to the real data of that year because the generation process uses non-DP validation results of the previous year.}
    \label{fig:utility-caps}
\end{figure}

Finally, EPrec of each request reported by researchers in the case study is summarised in \cref{tab:precision}. The average EPrec of the three datasets is .36, meaning that on average 36\% of findings derived from synthetic data were validated on real data. It is observed that the datasets exhibit different levels of EPrec, despite similar utility in terms of AJS. This implies that statistical utility metrics such as AJS may not capture real-world utility in downstream analyses.

\begin{table}
    \centering
    \caption{Reported epistemic precisions (EPrec).}
    \label{tab:precision}
    \footnotesize
    \begin{tabular}{lrrrrr}
        \toprule
        & \multicolumn{4}{c}{\textbf{Request counts}} \\
        \cmidrule(lr){2-5}
        \textbf{Dataset}& \textbf{EPrec: 0.0} & \textbf{0.5} & \textbf{1.0} & \textbf{All} & \textbf{EPrec} \\
        \midrule
        2022 & 3 & 5 & 1 & 9 & .39 \\
        2023 & 5 & 6 & 0 & 11 & .27 \\
        2024 & 0 & 5 & 0 & 5 & .50 \\
        \midrule
        \textbf{Total} & \textbf{8} & \textbf{16} & \textbf{1} & \textbf{25} & \textbf{avg. .36} \\
        \bottomrule
    \end{tabular}
\end{table}

\subsection{RQ2: Privacy risk of two-stage sharing}

For the privacy risk of the two-stage method, we first analyse empirical $\mu_{n,\sD}$-GWMIP. \cref{fig:mu-fit} shows the estimated $\mu$ at each request index in the chronological order for each dataset. These indicate empirical $\mu$-GWMIP with respect to $n=120$ students and the underlying distribution for the set of requests $1,\dots,R$. The regrets $\delta$ of the fit are moderate yet not sufficiently small ($M=0.0399$, $SD=0.0086$ for 2022, $M=0.0435$, $SD=0.0074$ for 2023 and $M=0.0485$, $SD=0.0105$ for dataset 2024), as Gomez et al.\ \cite{Gomez2025GDP} argue that $\delta < 0.01$ indicates a good fit. Thus our estimation might not fully capture the trade-off curves, which we aim to complement by reporting MIA advantages below. The red dashed lines indicate the $\mu$ values for $\mu$-GWMIP that is guaranteed by the DP-SDG process. Note that the DL-based method also guarantees the same level of formal privacy.

We observe that the privacy risk does not monotonically increase as the output size increases, perhaps because adding more output introduces noise as well as signal of membership. Overall, non-DP real-data validation introduces a moderate privacy risk to the DP-SDG baseline (i.e. the case of Stage 1 only). However, noting that non-DP SDG would offer no provable privacy guarantee and limited utility (\cref{fig:one-shot,fig:utility-caps}), the two-stage method significantly increases the practical utility by real-data validation as a trade-off with tolerable privacy leakage. Nonetheless, it should be noted that our attack is not optimal and only provide empirical risks, unlike the provable guarantee of DP.

\begin{figure*}[t]
    \centering
    \includegraphics[width=0.3\textwidth]{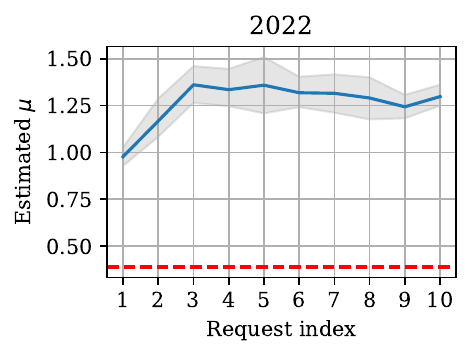}
    \includegraphics[width=0.3\textwidth]{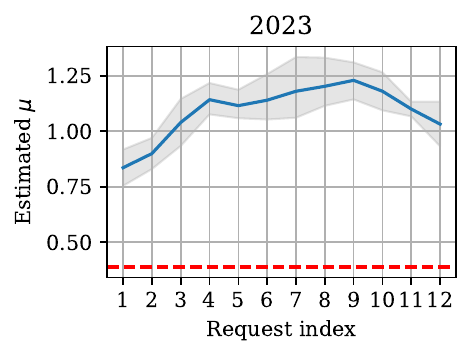} 
    \includegraphics[width=0.3\textwidth]{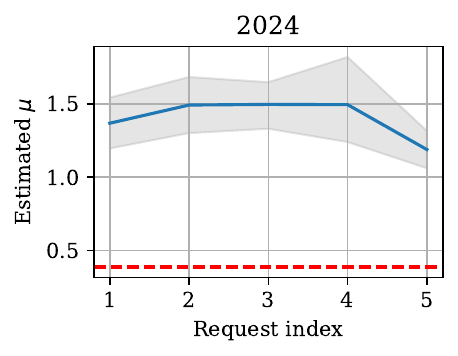}
    \caption{The empirical $\mu$-GWMIP with respect to $n=120$ and the underlying distribution for the two-stage sharing at each request index in the chronological order. $\mu$ at request $R$ accounts for DP-SDG and releasing the set of requests $1,\dots,R$. The red dashed lines indicate the privacy guarantee for which the DP-SDG alone provides without real-data validation (i.e. the case of Stage 1 only).}
    \label{fig:mu-fit}
\end{figure*}

To complement the empirical privacy parameter results, \cref{fig:advantage} shows the MIA advantage of the worst-case example with the set of requests $1,\dots,R$. While the dataset 2022 shows an explicitly increasing trend of MIA advantage with request counts, the other datasets exhibit less clear patterns. This indicates that privacy risk of real-data validation highly depends on the type of analysis, rather than the volume of output.

\begin{figure*}[t]
    \centering
    \includegraphics[width=0.3\textwidth]{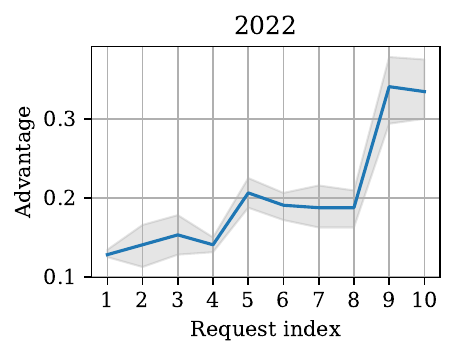}
    \includegraphics[width=0.3\textwidth]{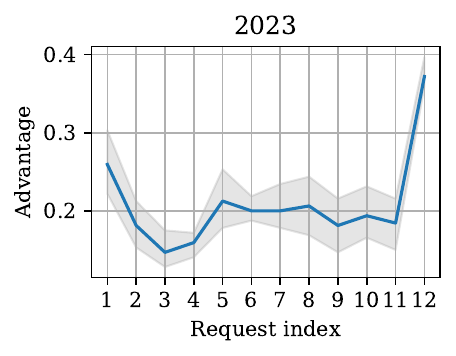} 
    \includegraphics[width=0.3\textwidth]{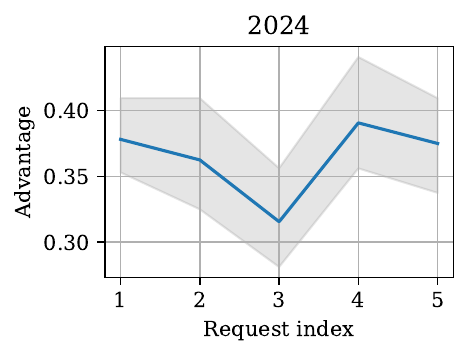}
    \caption{The MIA advantage of the worst-case example at each request index in the chronological order. The advantage at request $R$ is the MIA advantage when an attacker has access to the outputs of all requests $1,\dots,R$.}
    \label{fig:advantage}
\end{figure*}

To investigate how different types of secondary-use analysis pose different levels of privacy leakage, we plot the worst-case MIA advantage per request for each participant, as each participant consistently conducted similar analyses. \cref{fig:adv-requestor} shows that requests by researcher P4 exhibit remarkably high MIA advantage than those by the others. Qualitatively analysing submitted code, we find that P4 extensively focused on causal analysis while the others explored rather correlational aspects of the relationship between learning habits and learning outcomes. Specifically, P4 studied research questions such as ``what is the causal relationship between learning habits and learning outcome?'' using DirectLiNGAM \cite{Shimizu2011DirectLiNGAM}, a causal discovery technique. On the other hand, representative research questions of the other participants include ``How is the time-of-day patterns of learning habits related to learning outcomes?'' The discrepancy in privacy risk between these analysis types may be because DirectLiNGAM is sensitive to outliers \cite{Leyder2024TSLiNGAM} and thus extracts more signals of specific targets, increasing the MIA vulnerability.

\begin{figure}[tb]
    \centering
    \includegraphics[width=0.45\textwidth]{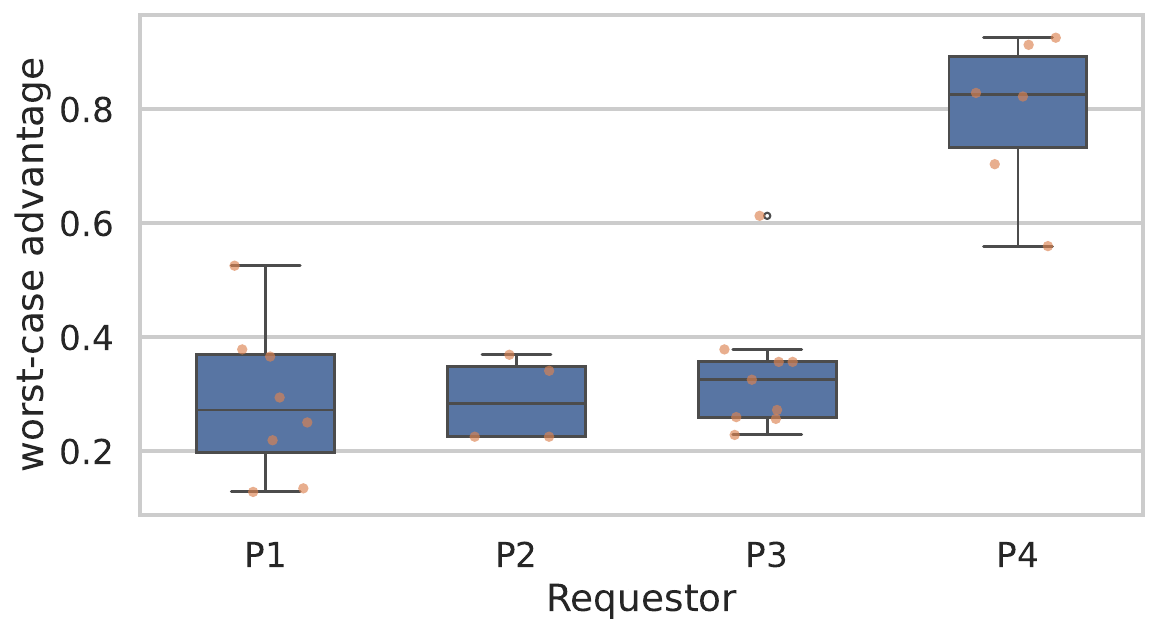}
    \caption{The worst-case MIA advantage per request for each participant.}
    \label{fig:adv-requestor}
\end{figure}

To further understand the relationship between the privacy risk and request characteristics, we show the pair-plot of EPrec, output dimension, the worst-case MIA advantage and empirical $\mu$-GWMIP, with respect to $n=120$ and the underlying distribution, of each request alone and highlight points that belong to P4. The request outputs of P4 exhibit relatively small dimensions and high privacy leakage, suggesting that the information content, rather than the output volume, primarily determines privacy risk. The negative correlation between output size and risk metrics also reflect this observation in that more output does not necessarily imply more privacy leakage. Additionally, we find that EPrec is negatively correlated with output size and positively correlated with risk metrics, though not significant, showing two important trade-offs. First, as a researcher requests for more output, the output more likely lead to false discoveries. Second, synthetic data with higher real-world utility (i.e. higher EPrec) tends to pose higher privacy leakage in the end within the two-stage sharing.

\begin{figure}[tb]
    \centering
    \includegraphics[width=0.48\textwidth]{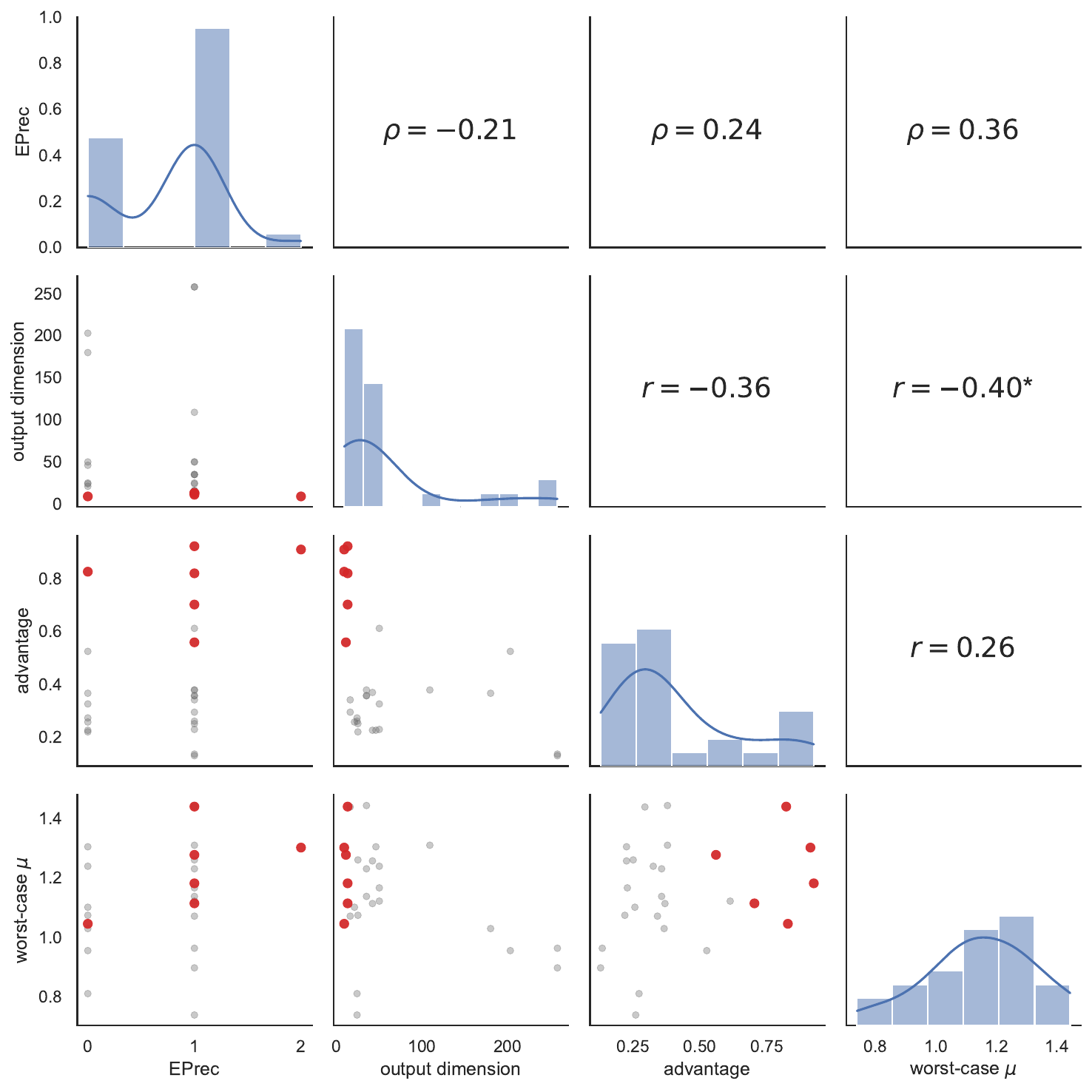}
    \caption{Pair-plot of epistemic precision (EPrec), output dimension, the worst-case MIA advantage and the empirical $\mu$-GWMIP with respect to $n=120$ and the underlying distribution. Points that belong to P4 are highlighted by red. For the pair of EPrec and another feature we calculate the Spearman correlation coefficient and otherwise the Pearson correlation. Statistical significance ($p<.05$) is indicated by an asterisk.}
    \label{fig:pairplot}
\end{figure}

\section{Discussion}

\subsection{Key findings}

While DP-SDG achieves robust, provable privacy guarantee, it remains largely infeasible for practical implementation in educational RWD sharing. To address this issue, we proposed a practical two-stage method with LLM-based DP-SDG and real-data validation. Though non-DP validation breaks the DP guarantee of Stage 1, we showed that privacy leakage of the entire data sharing process can still be analysed by focusing on WMIP. Nonetheless, it should again be noted that, while DP mechanisms provide formal guarantees, WMIP is necessarily empirical.

Regarding RQ1, the utility analysis demonstrated that our simple LLM-based DP-SDG shows comparable performance to the DL-based method while significantly reducing engineering costs. In addition, both methods were not significantly improved by adapting to the multi-shot setting, suggesting the need of further investigation of adaption techniques. The case study emulating a three-year secondary use of RWD demonstrated that the real-world utility, evaluated by EPrec, remains a challenge with only 36\% of synthetic findings validated on real data on average. Moreover, the inconsistency between the statistical metric of utility and EPrec indicates that common utility metrics do not necessarily capture the real-world utility of SDG methods.

For RQ2, the risk analysis revealed that real-data validation introduces measurable yet moderate privacy leakage. That is, the worst-case risk of membership inference and, consequently, attribute inference remains tolerable. Analysing accumulated risks over requests showed that allowing for more validation does not necessarily increase risks, potentially implying the need to carefully check the output of each request rather than minimising the number of requests. Further analysis demonstrated that analytical techniques with less statistical robustness, such as DirectLiNGAM, can pose higher privacy risk with a relatively small output size. We also observed the trade-off between EPrec and output size as well as EPrec and privacy. 

\subsection{Implications}

We consider implications of the above findings from the perspectives of data sharing research in the education domain, data custodians who employ the two-stage sharing and third-party researchers who analyse data shared by the two-stage method. 

First, for future research, our evaluation identified challenges in practical utility (i.e. EPrec) and risk mitigation, highlighting the need to further refine and enhance the proposed two-stage method. For example, the choice of summary statistics given to LLMs currently relies on heuristics, and the prompt for generating synthetic data remains simplistic and may not be optimal. As this work presented the potential of the simple, practical method with comparable performance to a DL-based method, more attention should be given to this line of work, rather than just applying state-of-the-art DP-SDG methods that are not practically usable most of the time in the field of education. In addition, though cyclically adapting private synthesis to multi-shot settings has the potential of improving generic, one-shot methods by taking the advantage of the iterative nature of education \cite{Ito2026CAPS}, our experiment did not demonstrate significant effectiveness. Thus, further study of multi-shot adaption would be needed. Importantly, as the case study disclosed, statistical metrics may not fully capture the real-world utility, thereby evaluation in more authentic settings being crucial. This aligns with the emerging recognition of real-world evaluation in the DP community \cite{Rosenblatt2023parity}. Furthermore, developing a methodology and clear criteria for output checking of real-data validation is needed based on risk analysis.

Second, the evaluation results help data custodians with communicating with stakeholders, including data subjects, institutions and third-party researchers. Since we focus on the worst-case privacy risk, the same privacy protection is undertaken for all individuals in the dataset. While it is hard to interpret and communicate privacy parameters like $\mu_{n,\sD}$ \cite{Dibia2025DPusability}, data custodians should focus on clarifying that the two-stage data sharing aims to minimise the risk that each student's participation in data sharing and particular learning behaviours are inferred by third parties \cite{Cyffers2025DPsetting}. For reviewing validation requests, our experimental results suggest that each analysis should be carefully checked, especially for its statistical robustness rather than simple heuristics like the output size or the number of requests, unlike traditional SDC techniques \cite{Griffiths2024SDC}. It should also be noted that, as we discussed in \cref{sec:two-stage}, WMIP does not imply protection against reconstruction attacks \cite{Salem2023privacygames}. Thus, the risk of dataset-level leakage should be taken into account depending on policies or regulations. 

Third, for researchers who perform secondary use of the shared data, our results help developing analysis by acknowledging the limitation of the synthetic data and privacy risk in real-data validation. As it is crucial for researchers to appreciate privacy protection for trustworthy secondary use \cite{Hubbard2020TRE}, a researcher should aim for non-disclosive output, avoiding analyses sensitive to outliers. The trade-off between EPrec and output-dimension should also be accounted in a validation request.

\subsection{Limitations}

The evaluation experiments relied on limited RWD focusing on learning habits of secondary-school students. The method should be tested for more various data for its utility and risks. We also relied on limited risk metrics where our attack may not be optimal, potentially underestimating the WMIP. Additionally, the case study is conducted with only four researchers with twenty-five requests in total, emulating three-year secondary use in a small scale. While this real-world evaluation is a significant first step beyond statistical metrics, the generalisability of the case study results may remain limited. Since our proposed method focuses on practicality, future work should test it more comprehensively in real-world settings.

\section{Conclusion}

Sharing educational RWD while protecting individual privacy is a substantial practical challenge. Though DP-SDG has gained growing attention for its theoretical promise, it remains far from being a viable solution in the education domain due to the locality of educational RWD. In this paper, we addressed this issue by presenting a novel two-stage method combining LLM-based DP-SDG and real-data validation. By focusing on the privacy notion of WMIP, we account for the risk of membership inference and attribute inference while allowing for real-data validation to support trustworthy secondary use and avoid false discoveries. Our empirical evaluation demonstrated that the simple LLM-based DP-SDG achieves comparable utility to the DL baseline with significantly reduced engineering costs, and that the entire two-stage sharing poses measurable but tolerable privacy risks. Nonetheless, the case study identified limited epistemic precision of synthetic findings. Overall, while the proposed method offers a highly practical method for sharing RWD in education, the evaluation suggests the need to further investigate risk mitigation and real-world utility enhancement.

\section*{Acknowledgments}
Anonymised.

{\appendices
\section{Details of the DL baseline}\label{sec:model-architecture}

For M1, both encoder and decoder have hidden layers of sizes 32 and 64. M2 has no hidden layer. To mitigate posterior collapse \cite{vanDenOord2017VQ}, we used dropout \cite{Srivastava2014Dropout} with probabilities 0.2 and 0.5 for the encoder and decoder, respectively. Additionally, cyclical $\beta$-annealing \cite{Fu2019annealing} was used for training M1 and M2 with maximum $\beta=0.001$. The ranges for HPO are shown in \cref{tab:hypers}.

\begin{table}[h!]
    \centering
    \caption{Ranges for hyperparameter optimisation}
    \label{tab:hypers}
    \begin{tabular}{@{}cll@{}}
        \toprule
        \multirow{2}{1.5em}{M1} & learning rate $\gamma$ & $[10^{-5}, 10^{-2}]$ (log-scale) \\
        & batch size & $[16, 512]$ ($\text{step}=8$) \\
        \midrule
        \multirow{5}{1.5em}{M2} & learning rate $\gamma$ & $[10^{-5}, 10^{-2}]$ (log-scale) \\
        & public batch size $B_\pub$ & $[16, 512]$ ($\text{step}=8$) \\
        & private batch size $B_\priv$ & $[1, |D_t|-1]$ ($\text{step}=1$) \\
        & epochs & $[1, 100]$ ($\text{step}=1$) \\
        & clipping norm $C$ & $[0.1, 5.0]$ ($\text{step}=0.1$) \\
        \bottomrule
    \end{tabular}
\end{table}

}

\bibliographystyle{IEEEtran}
\bibliography{references}

@ARTICLE{Mahajan2015RWD,
  title     = "{Real world data: Additional source for making clinical
               decisions}",
  author    = "Mahajan, Rajiv",
  journal   = "International journal of applied \& basic medical research",
  publisher = "Medknow",
  volume    =  5,
  number    =  2,
  pages     =  82,
  year      =  2015,
  optionalurl       = "http://dx.doi.org/10.4103/2229-516X.157148",
  doi       = "10.4103/2229-516X.157148"
}

@INPROCEEDINGS{Ito2026CAPS,
  title     = "{Cyclic Adaptive Private Synthesis for Sharing Real-World Data in
               Education}",
  author    = "Ito, Hibiki and Hsu, Chia-Yu and Ogata, Hirogaki",
  booktitle = "{Proceedings of the 16th International Learning Analytics and
               Knowledge Conference}",
  publisher = "Association for Computing Machinery",
  year      =  2026,
  optionalurl       = "http://dx.doi.org/10.1145/3785022.3785026",
  doi       = "10.1145/3785022.3785026"
}

@ARTICLE{Pan2024SRLA,
  title     = "{A systematic review of learning analytics: Incorporated
               instructional interventions on learning management systems}",
  author    = "Pan, Zilong and Biegley, Lauren and Taylor, Allen and Zheng, Hua",
  journal   = "Journal of Learning Analytics",
  publisher = "Society for Learning Analytics Research",
  volume    =  11,
  number    =  2,
  pages     = "52--72",
  year      =  2024,
  optionalurl       = "http://dx.doi.org/10.18608/jla.2023.8093",
  doi       = "10.18608/jla.2023.8093"
}

@ARTICLE{Okumura2026RWE,
  title     = "{Causal discovery for automated real-world educational evidence
               extraction}",
  author    = "Okumura, Koki and Nishioka, Kento and Koike, Kento and Horikoshi,
               Izumi and Ogata, Hiroaki",
  journal   = "Research and practice in technology enhanced learning",
  publisher = "Asia-Pacific Society for Computers in Education",
  volume    =  21,
  pages     =  020,
  year      =  2026,
  optionalurl       = "http://dx.doi.org/10.58459/rptel.2026.21020",
  doi       = "10.58459/rptel.2026.21020"
}

@ARTICLE{Kuromiya2023RWD,
  title     = "{Supporting reflective teaching workflow with real-world data and
               learning analytics}",
  author    = "Kuromiya, Hiroyuki and Nakanishi, Taro and Horikoshi, Izumi and
               Majumdar, Rwitajit and Ogata, Hiroaki",
  journal   = "Information and Technology in Education and Learning",
  publisher = "Japanese Society for Information and Systems in Education",
  volume    =  3,
  number    =  1,
  pages     = "Reg--p003",
  year      =  2023,
  optionalurl       = "https://www.jstage.jst.go.jp/article/itel/3/1/3_3.1.Reg.p003/_article",
  doi       = "10.12937/itel.3.1.reg.p003"
}

@INPROCEEDINGS{Hutt2023forgotten,
  title     = "{The right to be forgotten and educational data mining:
               Challenges and paths forward}",
  author    = "Hutt, Stephen and Das, Sanchari and Baker, Ryan",
  booktitle = "{Proceedings of the 16th International Conference on Educational
               Data Mining}",
  publisher = "Zenodo",
  pages     = "251--259",
  year      =  2023,
  optionalurl       = "https://educationaldatamining.org/edm2023/proceedings/2023.EDM-short-papers.23/2023.EDM-short-papers.23.pdf",
  doi       = "10.5281/ZENODO.8115655"
}

@INPROCEEDINGS{Baker2024open,
  title     = "{Open Science and Educational Data Mining: Which Practices Matter
               Most?}",
  author    = "Baker, Ryan S and Hutt, Stephen and Brooks, Christopher A and
               Srivastava, Namrata and Mills, Caitlin",
  booktitle = "{Proceedings of the 17th International Conference on Educational
               Data Mining}",
  publisher = "International Educational Data Mining Society",
  pages     = "279--287",
  year      =  2024,
  optionalurl       = "https://educationaldatamining.org/edm2024/proceedings/2024.EDM-short-papers.24/2024.EDM-short-papers.24.pdf",
  doi       = "10.5281/ZENODO.12729816"
}

@ARTICLE{Fischer2020bigdata,
  title     = "{Mining big data in education: Affordances and challenges}",
  author    = "Fischer, Christian and Pardos, Zachary A and Baker, Ryan Shaun
               and Williams, Joseph Jay and Smyth, Padhraic and Yu, Renzhe and
               Slater, Stefan and Baker, Rachel and Warschauer, Mark",
  journal   = "Review of research in education",
  publisher = "American Educational Research Association (AERA)",
  volume    =  44,
  number    =  1,
  pages     = "130--160",
  year      =  2020,
  optionalurl       = "https://journals.sagepub.com/doi/full/10.3102/0091732X20903304",
  doi       = "10.3102/0091732x20903304"
}

@INPROCEEDINGS{Dawson2019impact,
  title     = "{Increasing the impact of learning analytics}",
  author    = "Dawson, Shane and Joksimovic, Srecko and Poquet, Oleksandra and
               Siemens, George",
  booktitle = "{Proceedings of the 9th International Conference on Learning
               Analytics \& Knowledge}",
  publisher = "ACM",
  pages     = "446--455",
  series    = "LAK19",
  year      =  2019,
  optionalurl       = "https://doi.org/10.1145/3303772.3303784",
  doi       = "10.1145/3303772.3303784"
}

@INPROCEEDINGS{Samarati1998kanonymity,
  title     = "{Generalizing data to provide anonymity when disclosing
               information}",
  author    = "Samarati, Pierangela and Sweeney, Latanya",
  booktitle = "{Proceedings of the seventeenth ACM SIGACT-SIGMOD-SIGART
               symposium on Principles of database systems}",
  publisher = "ACM",
  pages     =  188,
  year      =  1998,
  optionalurl       = "https://www.eng.auburn.edu/~xqin/courses/comp7370/k-anonymity-pods98.pdf",
  doi       = "10.1145/275487.275508"
}

@INPROCEEDINGS{Machanavajjhala2006ldiversity,
  title     = "{L-diversity: privacy beyond k-anonymity}",
  author    = "Machanavajjhala, A and Gehrke, J and Kifer, D and
               Venkitasubramaniam, M",
  booktitle = "{22nd International Conference on Data Engineering (ICDE'06)}",
  publisher = "IEEE",
  pages     = "24--24",
  year      =  2006,
  optionalurl       = "https://ieeexplore.ieee.org/document/1617392",
  doi       = "10.1109/icde.2006.1"
}

@ARTICLE{Gadotti2024review,
  title     = "{Anonymization: The imperfect science of using data while
               preserving privacy}",
  author    = "Gadotti, Andrea and Rocher, Luc and Houssiau, Florimond and
               Cre\c{t}u, Ana-Maria and de Montjoye, Yves-Alexandre",
  journal   = "Science advances",
  publisher = "American Association for the Advancement of Science (AAAS)",
  volume    =  10,
  number    =  29,
  pages     = "eadn7053",
  year      =  2024,
  optionalurl       = "http://dx.doi.org/10.1126/sciadv.adn7053",
  doi       = "10.1126/sciadv.adn7053"
}

@INPROCEEDINGS{Cohen2022attacks,
  title     = "{Attacks on Deidentification's Defenses}",
  author    = "Cohen, A",
  booktitle = "{USENIX Secur Symp}",
  pages     = "1469--1486",
  year      =  2022,
  optionalurl       = "https://www.usenix.org/system/files/sec22-cohen.pdf"
}

@MISC{Ito20263PAE,
  title         = "{The third-party access effect: An overlooked challenge in
                   secondary use of educational real-world data}",
  author        = "Ito, Hibiki and Hsu, Chia-Yu and Ogata, Hiroaki",
  year          =  2026,
  optionalurl           = "http://arxiv.org/abs/2601.22472",
  archivePrefix = "arXiv",
  primaryClass  = "cs.CY",
  eprint        = "2601.22472",
  doi           = "10.48550/arXiv.2601.22472"
}

@INPROCEEDINGS{Ito2025detailed,
  title     = "{Too Detailed to Share? Towards Risk-Based Privacy Protection of
               Fine-Grained Educational Data}",
  author    = "Ito, Hibiki and Hsu, Chia-Yu and Ogata, Hiroaki",
  booktitle = "{The Proceedings of the 33rd International Conference on
               Computers in Education (ICCE 2025)}",
  pages     = "312--321",
  series    =  1,
  year      =  2025,
  optionalurl       = "https://library.apsce.net/index.php/ICCE/article/view/5958"
}

@INPROCEEDINGS{Dwork2006DP,
  title     = "{Calibrating Noise to Sensitivity in Private Data Analysis}",
  author    = "Dwork, Cynthia and McSherry, Frank and Nissim, Kobbi and Smith,
               Adam",
  editor    = "Halevi, Shai and Rabin, Tal",
  booktitle = "{Theory of Cryptography}",
  publisher = "Springer",
  pages     = "265--284",
  series    = "Lecture Notes in Computer Science",
  year      =  2006,
  optionalurl       = "http://dx.doi.org/10.1007/11681878_14",
  doi       = "10.1007/11681878\_14"
}

@INPROCEEDINGS{Annamalai2024attribute,
  title     = "{A Linear Reconstruction Approach for Attribute Inference Attacks
               against Synthetic Data}",
  author    = "Annamalai, Meenatchi Sundaram Muthu and Gadotti, Andrea and
               Rocher, Luc",
  booktitle = "{33rd USENIX Security Symposium (USENIX Security 24)}",
  pages     = "2351--2368",
  year      =  2024,
  optionalurl       = "https://www.usenix.org/system/files/usenixsecurity24-annamalai-linear.pdf"
}

@INPROCEEDINGS{Lu2019SDG,
  title     = "{Empirical evaluation on synthetic data generation with
               generative adversarial network}",
  author    = "Lu, Pei-Hsuan and Wang, Pang-Chieh and Yu, Chia-Mu",
  booktitle = "{Proceedings of the 9th International Conference on Web
               Intelligence, Mining and Semantics}",
  publisher = "ACM",
  year      =  2019,
  optionalurl       = "http://dx.doi.org/10.1145/3326467.3326474",
  doi       = "10.1145/3326467.3326474"
}

@INPROCEEDINGS{Stadler2022groundhog,
  title     = "{Synthetic Data -- Anonymisation Groundhog Day}",
  author    = "Stadler, Theresa and Oprisanu, Bristena and Troncoso, Carmela",
  booktitle = "{31st USENIX Security Symposium (USENIX Security 22)}",
  pages     = "1451--1468",
  year      =  2022,
  optionalurl       = "https://www.usenix.org/conference/usenixsecurity22/presentation/stadler"
}

@ARTICLE{vanDerLinden2025synthhealth,
  title     = "{Who needs real data anyway? Exploring the use of synthetic data
               in economic evaluations of health interventions}",
  author    = "van der Linden, N and Pouwels, X G L V and Jahn, B and Siebert, U
               and Koffijberg, H",
  journal   = "Value in health: the journal of the International Society for
               Pharmacoeconomics and Outcomes Research",
  publisher = "Elsevier BV",
  year      =  2025,
  optionalurl       = "http://dx.doi.org/10.1016/j.jval.2025.06.007",
  doi       = "10.1016/j.jval.2025.06.007"
}

@ARTICLE{MontoyaPerez2024FP,
  title     = "{Does differentially private synthetic data lead to synthetic
               discoveries?}",
  author    = "Montoya Perez, Ileana and Movahedi, Parisa and Nieminen, Valtteri
               and Airola, Antti and Pahikkala, Tapio",
  journal   = "Methods of Information in Medicine",
  publisher = "Georg Thieme Verlag KG",
  volume    =  63,
  number    = "1-02",
  pages     = "35--51",
  year      =  2024,
  optionalurl       = "http://www.thieme-connect.de/DOI/DOI?10.1055/a-2385-1355",
  doi       = "10.1055/a-2385-1355"
}

@ARTICLE{Nguyen2024small,
  title     = "{Learning analytics with small datasets--state of the art and
               beyond}",
  author    = "Nguyen, Ngoc Buu Cat and Karunaratne, Thashmee",
  journal   = "Education Sciences",
  publisher = "MDPI AG",
  volume    =  14,
  number    =  6,
  pages     =  608,
  year      =  2024,
  optionalurl       = "http://dx.doi.org/10.3390/educsci14060608",
  doi       = "10.3390/educsci14060608"
}

@INPROCEEDINGS{Medeiros2025datagov,
  title     = "{Data governance in education: Addressing challenges and
               unlocking opportunities for effective data management}",
  author    = "Medeiros, Thiago and Ara\'{u}jo, Andr\'{e} and Silva, Jos\'{e}
               and Silva, Alenilton",
  booktitle = "{Proceedings of the 27th International Conference on Enterprise
               Information Systems}",
  publisher = "SCITEPRESS - Science and Technology Publications",
  pages     = "367--374",
  year      =  2025,
  optionalurl       = "http://dx.doi.org/10.5220/0013468300003929",
  doi       = "10.5220/0013468300003929"
}

@INCOLLECTION{Ullman2011hardness,
  title     = "{PCPs and the hardness of generating private synthetic data}",
  author    = "Ullman, Jonathan and Vadhan, Salil",
  booktitle = "{Theory of Cryptography}",
  publisher = "Springer Berlin Heidelberg",
  pages     = "400--416",
  series    = "Lecture notes in computer science",
  year      =  2011,
  optionalurl       = "http://dx.doi.org/10.1007/978-3-642-19571-6_24",
  doi       = "10.1007/978-3-642-19571-6\_24"
}

@MASTERSTHESIS{Ito2024thesis,
  title     = "{A Data-Centric Analysis of Membership Inference Attacks}",
  author    = "Ito, Hibiki",
  school    = "University of Helsinki",
  publisher = "University of Helsinki",
  year      =  2024,
  optionalurl       = "http://hdl.handle.net/10138/589286"
}

@ARTICLE{Pardo2014ethical,
  title     = "{Ethical and privacy principles for learning analytics: Ethical
               and privacy principles}",
  author    = "Pardo, Abelardo and Siemens, George",
  journal   = "British journal of educational technology: journal of the Council
               for Educational Technology",
  publisher = "Wiley",
  volume    =  45,
  number    =  3,
  pages     = "438--450",
  year      =  2014,
  optionalurl       = "https://onlinelibrary.wiley.com/doi/abs/10.1111/bjet.12152",
  doi       = "10.1111/bjet.12152"
}

@ARTICLE{Liu2023review,
  title     = "{Understanding privacy and data protection issues in learning
               analytics using a systematic review}",
  author    = "Liu, Qinyi and Khalil, Mohammad",
  journal   = "British Journal of Educational Technology: Journal of the Council
               for Educational Technology",
  publisher = "Wiley",
  volume    =  54,
  number    =  6,
  pages     = "1715--1747",
  year      =  2023,
  optionalurl       = "https://onlinelibrary.wiley.com/doi/abs/10.1111/bjet.13388",
  doi       = "10.1111/bjet.13388"
}

@INPROCEEDINGS{Wijerathne2024EREDA,
  title     = "{Empowering educational researchers with a privacy-centric data
               platform: Design, implementation, and implications}",
  author    = "Wijerathne, Isanka and Flanagan, Brendan and Ogata, Hiroaki",
  booktitle = "{Proceedings of the 32nd International Conference on Computers in
               Education}",
  publisher = "Asia-Pacific Society for Computers in Education",
  year      =  2024,
  optionalurl       = "https://library.apsce.net/index.php/ICCE/article/view/4878",
  doi       = "10.58459/icce.2024.4878"
}

@INPROCEEDINGS{Flanagana2018infra,
  title     = "{Integration of learning analytics research and production
               systems while protecting privacy}",
  author    = "Flanagana, Brendan and Ogata, Hiroaki",
  booktitle = "{Proceedings of the 25th International Conference on Computers in
               Education}",
  year      =  2018,
  optionalurl       = "https://library.apsce.net/index.php/ICCE/article/view/2103"
}

@INCOLLECTION{Joksimovic2022privacy,
  title     = "{Privacy-driven learning analytics}",
  author    = "Joksimovi\'{c}, Sre\'{c}ko and Marshall, Ruth and Rakotoarivelo,
               Thierry and Ladjal, Djazia and Zhan, Chen and Pardo, Abelardo",
  booktitle = "{Manage Your Own Learning Analytics}",
  publisher = "Springer International Publishing",
  pages     = "1--22",
  series    = "Smart Innovation, Systems and Technologies",
  year      =  2022,
  optionalurl       = "https://link.springer.com/chapter/10.1007/978-3-030-86316-6_1",
  doi       = "10.1007/978-3-030-86316-6\_1"
}

@INPROCEEDINGS{Torre2020ELAT,
  title     = "{edX log data analysis made easy: introducing ELAT: An
               open-source, privacy-aware and browser-based edX log data
               analysis tool}",
  author    = "Torre, Manuel Valle and Tan, Esther and Hauff, Claudia",
  booktitle = "{Proceedings of the Tenth International Conference on Learning
               Analytics \& Knowledge}",
  publisher = "ACM",
  year      =  2020,
  optionalurl       = "http://dx.doi.org/10.1145/3375462.3375510",
  doi       = "10.1145/3375462.3375510"
}

@ARTICLE{Yacobson2021deident,
  title     = "{De-identification is insufficient to protect student privacy, or
               -- what can a field trip reveal?}",
  author    = "Yacobson, Elad and Fuhrman, Orly and Hershkovitz, Sara and
               Alexandron, Giora",
  journal   = "Journal of learning analytics",
  publisher = "Society for Learning Analytics Research",
  volume    =  8,
  number    =  2,
  pages     = "83--92",
  year      =  2021,
  optionalurl       = "https://www.learning-analytics.info/index.php/JLA/article/view/7353",
  doi       = "10.18608/jla.2021.7353"
}

@ARTICLE{Vatsalan2022risk,
  title     = "{Privacy risk quantification in education data using Markov
               model}",
  author    = "Vatsalan, Dinusha and Rakotoarivelo, Thierry and Bhaskar, Raghav
               and Tyler, Paul and Ladjal, Djazia",
  journal   = "British journal of educational technology: journal of the Council
               for Educational Technology",
  publisher = "Wiley",
  volume    =  53,
  number    =  4,
  pages     = "804--821",
  year      =  2022,
  optionalurl       = "https://onlinelibrary.wiley.com/doi/abs/10.1111/bjet.13223",
  doi       = "10.1111/bjet.13223"
}

@ARTICLE{Kyritsi2019privacy,
  title    = "{The Pursuit of Patterns in Educational Data Mining as a Threat to
              Student Privacy}",
  author   = "Kyritsi, Kyriaki H and Zorkadis, Vassilios and Stavropoulos, Elias
              C and Verykios, Vassilios S",
  journal  = "Journal of Interactive Media in Education",
  volume   =  2019,
  number   =  1,
  pages    =  2,
  year     =  2019,
  optionalurl      = "https://jime.open.ac.uk/articles/10.5334/jime.502",
  doi      = "10.5334/jime.502"
}

@INPROCEEDINGS{Angiuli2016kanon,
  title     = "{Statistical tradeoffs between generalization and suppression in
               the DE-identification of large-scale data sets}",
  author    = "Angiuli, Olivia and Waldo, Jim",
  booktitle = "{2016 IEEE 40th Annual Computer Software and Applications
               Conference (COMPSAC)}",
  publisher = "IEEE",
  volume    =  2,
  pages     = "589--593",
  year      =  2016,
  optionalurl       = "https://ieeexplore.ieee.org/abstract/document/7552278",
  doi       = "10.1109/compsac.2016.198"
}

@ARTICLE{Stinar2024kanon,
  title     = "{An approach to improve k-anonymization practices in educational
               data mining}",
  author    = "Stinar, Frank and Xiong, Zihan and Bosch, Nigel",
  journal   = "Journal of Educational Data Mining",
  publisher = "Zenodo",
  volume    =  16,
  number    =  1,
  pages     = "61--83",
  year      =  2024,
  optionalurl       = "https://jedm.educationaldatamining.org/index.php/JEDM/article/view/764",
  doi       = "10.5281/ZENODO.11056083"
}

@ARTICLE{Kuzilek2017OULAD,
  title    = "{Open University Learning Analytics dataset}",
  author   = "Kuzilek, Jakub and Hlosta, Martin and Zdrahal, Zdenek",
  journal  = "Scientific Data",
  volume   =  4,
  number   =  1,
  pages    =  170171,
  year     =  2017,
  optionalurl      = "https://www.nature.com/articles/sdata2017171",
  doi      = "10.1038/sdata.2017.171"
}

@MISC{HarvardX2014data,
  title     = "{HarvardX Person-Course Academic Year 2013 De-Identified dataset,
               version 3.0}",
  author    = "{HarvardX}",
  publisher = "Harvard Dataverse",
  year      =  2014,
  optionalurl       = "http://dx.doi.org/10.7910/DVN/26147",
  doi       = "10.7910/DVN/26147"
}

@ARTICLE{Prasser2020ARX,
  title     = "{Flexible data anonymization using ARX--Current status and
               challenges ahead: Flexible data anonymization using ARX-Current
               status and challenges ahead}",
  author    = "Prasser, Fabian and Eicher, Johanna and Spengler, Helmut and
               Bild, Raffael and Kuhn, Klaus A",
  journal   = "Software: practice \& experience",
  publisher = "Wiley",
  volume    =  50,
  number    =  7,
  pages     = "1277--1304",
  year      =  2020,
  optionalurl       = "http://dx.doi.org/10.1002/spe.2812",
  doi       = "10.1002/spe.2812"
}

@MISC{CSIRO2019R4,
  title        = "{Re-Identification Risk Quantification}",
  author       = "{CSIRO}",
  booktitle    = "{Privacy Technology Research Group}",
  year         =  2019,
  howpublished = "\url{https://research.csiro.au/isp/research/privacy/r4/}"
}

@ARTICLE{Fredrikson2014AI,
  title    = "{Privacy in pharmacogenetics: An end-to-end case study of
              personalized warfarin dosing}",
  author   = "Fredrikson, Matt and Lantz, Eric and Jha, S and Lin, Simon M and
              Page, David and Ristenpart, Thomas",
  journal  = "USENIX Security Symposium",
  volume   =  2014,
  pages    = "17--32",
  year     =  2014,
  optionalurl      = "https://www.usenix.org/system/files/conference/usenixsecurity14/sec14-paper-fredrikson-privacy.pdf"
}

@INPROCEEDINGS{Shokri2017MIA,
  title     = "{Membership inference attacks against machine learning models}",
  author    = "Shokri, Reza and Stronati, Marco and Song, Congzheng and
               Shmatikov, Vitaly",
  booktitle = "{2017 IEEE Symposium on Security and Privacy (SP)}",
  publisher = "IEEE",
  pages     = "3--18",
  year      =  2017,
  optionalurl       = "https://ieeexplore.ieee.org/abstract/document/7958568",
  doi       = "10.1109/sp.2017.41"
}

@ARTICLE{Flanagan2022synth,
  title     = "{Fine grain synthetic educational data: Challenges and
               limitations of collaborative learning analytics}",
  author    = "Flanagan, Brendan and Majumdar, Rwitajit and Ogata, Hiroaki",
  journal   = "IEEE access: practical innovations, open solutions",
  publisher = "Institute of Electrical and Electronics Engineers (IEEE)",
  volume    =  10,
  pages     = "26230--26241",
  year      =  2022,
  optionalurl       = "https://ieeexplore.ieee.org/document/9726239",
  doi       = "10.1109/access.2022.3156073"
}

@INPROCEEDINGS{Khalil2025CTGAN,
  title     = "{Creating artificial students that never existed: Leveraging
               large language models and CTGANs for synthetic data generation}",
  author    = "Khalil, Mohammad and Vadiee, Farhad and Shakya, Ronas and Liu,
               Qinyi",
  booktitle = "{Proceedings of the 15th International Learning Analytics and
               Knowledge Conference}",
  publisher = "ACM",
  pages     = "439--450",
  year      =  2025,
  optionalurl       = "http://dx.doi.org/10.1145/3706468.3706523",
  doi       = "10.1145/3706468.3706523"
}

@INCOLLECTION{Vie2022synth,
  title     = "{Privacy-preserving synthetic educational data generation}",
  author    = "Vie, Jill-J\^{e}nn and Rigaux, Tomas and Minn, Sein",
  booktitle = "{Lecture Notes in Computer Science}",
  publisher = "Springer International Publishing",
  pages     = "393--406",
  series    = "Lecture notes in computer science",
  year      =  2022,
  optionalurl       = "https://link.springer.com/chapter/10.1007/978-3-031-16290-9_29",
  doi       = "10.1007/978-3-031-16290-9\_29"
}

@ARTICLE{Iloh2025CTGAN,
  title     = "{Generative private synthetic student data for learning
               analytics: An empirical study}",
  author    = "Iloh, Divine and Olayinka, Oluwakemi Temitope and Iyere, Faith
               and Chilakala, Sanath",
  journal   = "IEEE access: practical innovations, open solutions",
  publisher = "Institute of Electrical and Electronics Engineers (IEEE)",
  volume    = "PP",
  number    =  99,
  pages     = "1--1",
  year      =  2025,
  optionalurl       = "http://dx.doi.org/10.1109/ACCESS.2025.3619091",
  doi       = "10.1109/access.2025.3619091"
}

@INCOLLECTION{Bautista2021GAN,
  title     = "{Protecting student privacy with synthetic data from generative
               adversarial networks}",
  author    = "Bautista, Peter and Inventado, Paul Salvador",
  booktitle = "{Lecture Notes in Computer Science}",
  publisher = "Springer International Publishing",
  pages     = "66--70",
  series    = "Lecture notes in computer science",
  year      =  2021,
  optionalurl       = "https://link.springer.com/chapter/10.1007/978-3-030-78270-2_11",
  doi       = "10.1007/978-3-030-78270-2\_11"
}

@INPROCEEDINGS{Liu2025fairprivate,
  title     = "{Can synthetic data be fair and private? A comparative study of
               synthetic data generation and fairness algorithms}",
  author    = "Liu, Qinyi and Deho, Oscar and Vadiee, Farhad and Khalil,
               Mohammad and Joksimovic, Srecko and Siemens, George",
  booktitle = "{Proceedings of the 15th International Learning Analytics and
               Knowledge Conference}",
  publisher = "ACM",
  pages     = "591--600",
  year      =  2025,
  optionalurl       = "https://doi.org/10.1145/3706468.3706546",
  doi       = "10.1145/3706468.3706546"
}

@ARTICLE{Liu2025DPGAN,
  title     = "{Ensuring privacy through synthetic data generation in education}",
  author    = "Liu, Qinyi and Shakya, Ronas and Jovanovic, Jelena and Khalil,
               Mohammad and de la Hoz-Ruiz, Javier",
  journal   = "British journal of educational technology: journal of the Council
               for Educational Technology",
  publisher = "Wiley",
  volume    =  56,
  number    =  3,
  pages     = "1053--1073",
  year      =  2025,
  optionalurl       = "https://onlinelibrary.wiley.com/doi/abs/10.1111/bjet.13576",
  doi       = "10.1111/bjet.13576"
}

@ARTICLE{Zhan2024DPGAN,
  title     = "{Preserving both privacy and utility in learning analytics}",
  author    = "Zhan, Chen and Joksimovi\'{c}, Sre\'{c}ko and Ladjal, Djazia and
               Rakotoarivelo, Thierry and Marshall, Ruth and Pardo, Abelardo",
  journal   = "IEEE transactions on learning technologies",
  publisher = "Institute of Electrical and Electronics Engineers (IEEE)",
  volume    =  17,
  pages     = "1615--1627",
  year      =  2024,
  optionalurl       = "https://ieeexplore.ieee.org/document/10508498",
  doi       = "10.1109/tlt.2024.3393766"
}

@ARTICLE{Kesgin2025FairSYN,
  title     = "{FairSYN-Edu a diffusion-based model for fair and private
               educational data synthesis}",
  author    = "Kesgin, Kadir",
  journal   = "Discover education",
  publisher = "Springer Science and Business Media LLC",
  volume    =  4,
  number    =  1,
  pages     = "1--18",
  year      =  2025,
  optionalurl       = "http://dx.doi.org/10.1007/s44217-025-00743-9",
  doi       = "10.1007/s44217-025-00743-9"
}

@INCOLLECTION{Saqr2026idio,
  title     = "{Individualized analytics: Within-person and idiographic
               analysis}",
  author    = "Saqr, Mohammed and Ito, Hibiki and L\'{o}pez-Pernas, Sonsoles",
  booktitle = "{Advanced Learning Analytics Methods}",
  publisher = "Springer Nature Switzerland",
  pages     = "471--491",
  year      =  2026,
  optionalurl       = "https://lamethods.org/book2/chapters/ch18-idio/ch18-idio.html",
  doi       = "10.1007/978-3-031-95365-1\_18"
}

@ARTICLE{Mao2024DPtime,
  title   = "{Differential Privacy for Time Series: A Survey}",
  author  = "Mao, Yulian and Ye, Qingqing and Wang, Qi and Hu, Haibo",
  journal = "IEEE Bulletin of the Technical Committee on Data Engineering",
  volume  =  48,
  number  =  1,
  pages   = "67--92",
  year    =  2024,
  optionalurl     = "https://openreview.net/forum?id=D09lkXhXoc"
}

@ARTICLE{He2024onlineSDG,
  title     = "{Online Differentially Private Synthetic Data Generation}",
  author    = "He, Yiyun and Vershynin, Roman and Zhu, Yizhe",
  journal   = "IEEE transactions on privacy",
  publisher = "Institute of Electrical and Electronics Engineers (IEEE)",
  volume    =  1,
  pages     = "19--30",
  year      =  2024,
  optionalurl       = "http://dx.doi.org/10.1109/tp.2024.3486687",
  doi       = "10.1109/tp.2024.3486687"
}

@ARTICLE{Bun2024continual,
  title     = "{Continual release of differentially private synthetic data from
               longitudinal data collections}",
  author    = "Bun, Mark and Gaboardi, Marco and Neunhoeffer, Marcel and Zhang,
               Wanrong",
  journal   = "Proceedings of the ACM on management of data",
  publisher = "Association for Computing Machinery (ACM)",
  volume    =  2,
  number    =  2,
  pages     = "1--26",
  year      =  2024,
  optionalurl       = "http://dx.doi.org/10.1145/3651595",
  doi       = "10.1145/3651595"
}

@MISC{Gu2025continuous,
  title         = "{Private continuous-time synthetic trajectory generation via
                   mean-field Langevin dynamics}",
  author        = "Gu, Anming and Chien, Edward and Greenewald, Kristjan",
  year          =  2025,
  optionalurl           = "http://arxiv.org/abs/2506.12203",
  archivePrefix = "arXiv",
  primaryClass  = "cs.LG",
  eprint        = "2506.12203",
  doi           = "10.48550/arXiv.2506.12203"
}

@ARTICLE{Wang2025DPTrajPM,
  title     = "{DPTraj-PM: Differentially private trajectory synthesis using
               prefix tree and Markov process}",
  author    = "Wang, Nana and Kankanhalli, Mohan",
  journal   = "ACM Transactions on Spatial Algorithms and Systems",
  publisher = "Association for Computing Machinery (ACM)",
  number    =  3778164,
  year      =  2025,
  optionalurl       = "http://dx.doi.org/10.1145/3778164",
  doi       = "10.1145/3778164"
}

@INPROCEEDINGS{He2023PSMM,
  title     = "{Algorithmically Effective Differentially Private Synthetic Data}",
  author    = "He, Yiyun and Vershynin, Roman and Zhu, Yizhe",
  booktitle = "{The Thirty Sixth Annual Conference on Learning Theory}",
  publisher = "PMLR",
  pages     = "3941--3968",
  year      =  2023,
  optionalurl       = "https://proceedings.mlr.press/v195/he23a.html"
}

@INPROCEEDINGS{Wang2025PEimagetext,
  title     = "{Synthesize Privacy-Preserving High-Resolution Images via Private
               Textual Intermediaries}",
  author    = "Wang, Haoxiang and Lin, Zinan and Yu, Da and Zhang, Huishuai",
  booktitle = "{The Thirty-ninth Annual Conference on Neural Information
               Processing Systems}",
  year      =  2025,
  optionalurl       = "https://openreview.net/forum?id=g8zr9rxRHm"
}

@INPROCEEDINGS{Lin2025PEsim,
  title     = "{Differentially Private Synthetic Data via APIs 3: Using
               Simulators Instead of Foundation Model}",
  author    = "Lin, Zinan and Baltrusaitis, Tadas and Yekhanin, Sergey",
  booktitle = "{Will Synthetic Data Finally Solve the Data Access Problem?}",
  year      =  2025,
  optionalurl       = "https://openreview.net/forum?id=7AvUO96Pmk"
}

@INPROCEEDINGS{Swanberg2025PEtabular,
  title     = "{Is API Access to LLMs Useful for Generating Private Synthetic
               Tabular Data?}",
  author    = "Swanberg, Marika and McKenna, Ryan and Roth, Edo and Cheu, Albert
               and Kairouz, Peter",
  booktitle = "{Will Synthetic Data Finally Solve the Data Access Problem?}",
  year      =  2025,
  optionalurl       = "https://openreview.net/forum?id=e4uwE1muhR"
}

@INPROCEEDINGS{Zhang2025PCEvolve,
  title     = "{PCEvolve: Private Contrastive Evolution for Synthetic Dataset
               Generation via Few-Shot Private Data and Generative {APIs}}",
  author    = "Zhang, Jianqing and Liu, Yang and Fu, Jie and Hua, Yang and Zou,
               Tianyuan and Cao, Jian and Yang, Qiang",
  booktitle = "{Forty-second International Conference on Machine Learning}",
  year      =  2025,
  optionalurl       = "https://openreview.net/forum?id=IKCfxWtTsu$\lnot{}$eId=UJypmbyVvk"
}

@INPROCEEDINGS{Lin2023PEimage,
  title     = "{Differentially Private Synthetic Data via Foundation Model APIs
               1: Images}",
  author    = "Lin, Zinan and Gopi, Sivakanth and Kulkarni, Janardhan and Nori,
               Harsha and Yekhanin, Sergey",
  booktitle = "{The Twelfth International Conference on Learning
               Representations}",
  year      =  2023,
  optionalurl       = "https://openreview.net/pdf?id=YEhQs8POIo"
}

@INPROCEEDINGS{Gonzalez2025PEconverge,
  title     = "{Private Evolution Converges}",
  author    = "Gonz\'{a}lez, Tom\'{a}s and Fanti, Giulia and Ramdas, Aaditya",
  booktitle = "{The Thirty-ninth Annual Conference on Neural Information
               Processing Systems}",
  year      =  2025,
  optionalurl       = "https://openreview.net/forum?id=zOCENGh1Jg"
}

@INPROCEEDINGS{Xie2024PEtext,
  title     = "{Differentially Private Synthetic Data via Foundation Model APIs
               2: Text}",
  author    = "Xie, Chulin and Lin, Zinan and Backurs, Arturs and Gopi,
               Sivakanth and Yu, Da and Inan, Huseyin A and Nori, Harsha and
               Jiang, Haotian and Zhang, Huishuai and Lee, Yin Tat and Li, Bo
               and Yekhanin, Sergey",
  booktitle = "{International Conference on Machine Learning}",
  publisher = "PMLR",
  pages     = "54531--54560",
  year      =  2024,
  optionalurl       = "https://proceedings.mlr.press/v235/xie24g.html"
}

@ARTICLE{Weise2024TRE,
  title     = "{Trusted Research Environments: Analysis of characteristics and
               data availability}",
  author    = "Weise, Martin and Rauber, Andreas",
  journal   = "International Journal of Digital Curation",
  publisher = "University of Edinburgh",
  volume    =  18,
  number    =  1,
  year      =  2024,
  optionalurl       = "http://dx.doi.org/10.2218/ijdc.v18i1.939",
  doi       = "10.2218/ijdc.v18i1.939"
}

@MISC{USCensusBureau2018SIPP,
  title        = "{The Creation and Use of the SIPP Synthetic Beta v7.0}",
  author       = "{US Census Bureau}",
  booktitle    = "{Census.gov}",
  year         =  2018,
  howpublished = "\url{https://www.census.gov/library/working-papers/2018/adrm/SIPP-Synthetic-Beta.html}"
}

@INCOLLECTION{Drechsler2014synthLB,
  title     = "{Synthetic longitudinal business databases for international
               comparisons}",
  author    = "Drechsler, J{\"{o}}rg and Vilhuber, Lars",
  booktitle = "{Privacy in Statistical Databases}",
  publisher = "Springer International Publishing",
  pages     = "243--252",
  series    = "Lecture Notes in Computer Science",
  year      =  2014,
  optionalurl       = "http://dx.doi.org/10.1007/978-3-319-11257-2_19",
  doi       = "10.1007/978-3-319-11257-2\_19"
}

@INPROCEEDINGS{Ye2022MIA,
  title     = "{Enhanced membership inference attacks against machine learning
               models}",
  author    = "Ye, Jiayuan and Maddi, Aadyaa and Murakonda, Sasi Kumar and
               Bindschaedler, Vincent and Shokri, Reza",
  booktitle = "{Proceedings of the 2022 ACM SIGSAC Conference on Computer and
               Communications Security}",
  publisher = "ACM",
  year      =  2022,
  optionalurl       = "https://dl.acm.org/doi/10.1145/3548606.3560675",
  doi       = "10.1145/3548606.3560675"
}

@ARTICLE{Dong2022GDP,
  title     = "{Gaussian differential privacy}",
  author    = "Dong, Jinshuo and Roth, Aaron and Su, Weijie J",
  journal   = "Journal of the Royal Statistical Society. Series B, Statistical
               methodology",
  publisher = "Oxford University Press (OUP)",
  volume    =  84,
  number    =  1,
  pages     = "3--37",
  year      =  2022,
  optionalurl       = "http://dx.doi.org/10.1111/rssb.12454",
  doi       = "10.1111/rssb.12454"
}

@INPROCEEDINGS{Gaboardi2020OpenDP,
  title     = "{A Programming Framework for {OpenDP}}",
  author    = "Gaboardi, Marco and Hay, Michael and Vadhan, Salil",
  booktitle = "{6th Workshop on the Theory and Practice of Differential Privacy
               (TPDP 2020)}",
  year      =  2020,
  optionalurl       = "https://privacytools.seas.harvard.edu/sites/g/files/omnuum6656/files/privacytools/files/opendp_programming_framework_11.pdf"
}

@INPROCEEDINGS{Hsu2023habitsmining,
  title     = "{Learning Habits Mining and data-driven support of building
               habits in education}",
  author    = "Hsu, Chia-Yu and Horikoshi, Izumi and Majumdar, Rwitajit and
               Ogata, Hiroaki",
  booktitle = "{Proceedings of the 31st International Conference on Computers in
               Education}",
  publisher = "Asia-Pacific Society for Computers in Education",
  year      =  2023,
  optionalurl       = "https://library.apsce.net/index.php/ICCE/article/view/4783",
  doi       = "10.58459/icce.2023.4783"
}

@ARTICLE{ShirvaniBoroujeni2019MOOChabits,
  title     = "{Discovery and temporal analysis of MOOC study patterns}",
  author    = "Shirvani Boroujeni, Mina and Dillenbourg, Pierre",
  journal   = "Journal of learning analytics",
  publisher = "Society for Learning Analytics Research",
  volume    =  6,
  number    =  1,
  pages     = "16--33",
  year      =  2019,
  optionalurl       = "http://dx.doi.org/10.18608/jla.2019.61.2",
  doi       = "10.18608/jla.2019.61.2"
}

@ARTICLE{Magulod2019habits,
  title     = "{Learning styles, study habits and academic performance of
               Filipino University students in applied science courses:
               Implications for instruction}",
  author    = "Magulod, Gilbert",
  journal   = "Journal of technology and science education",
  publisher = "Omnia Publisher SL",
  volume    =  9,
  number    =  2,
  pages     =  184,
  year      =  2019,
  optionalurl       = "http://dx.doi.org/10.3926/jotse.504",
  doi       = "10.3926/jotse.504"
}

@ARTICLE{Ricker2020habits,
  title     = "{Student clickstream data: Does time of day matter?}",
  author    = "Ricker, Gina and Koziarski, Mathew and Walters, Alyssa",
  journal   = "Journal of Online Learning Research",
  publisher = "Association for the Advancement of Computing in Education. P.O.
               Box 719, Waynesville, NC 28786. Tel: 828-246-9558; Fax:
               828-246-9557; e-mail: info@aace.org; Web site:
               https://www.aace.org/pubs/jolr/",
  volume    =  6,
  number    =  2,
  pages     = "155--170",
  year      =  2020,
  optionalurl       = "https://eric.ed.gov/?id=EJ1273645"
}

@INCOLLECTION{Hsu2024productivity,
  title     = "{Evaluating productivity of learning habits using math learning
               logs: Do K12 learners manage their time effectively?}",
  author    = "Hsu, Chia-Yu and Horikoshi, Izumi and Li, Huiyong and Majumdar,
               Rwitajit and Ogata, Hiroaki",
  editor    = "Ferreira, Mello and {Rafael} and Rummel, Nikol and Jivet, Ioana
               and Pishtari, Gerti and Ruip\'{e}rez Valiente, Jos\'{e} A",
  booktitle = "{Technology Enhanced Learning for Inclusive and Equitable Quality
               Education}",
  publisher = "Springer Nature Switzerland",
  volume    =  15159,
  pages     = "168--178",
  year      =  2024,
  optionalurl       = "https://link.springer.com/chapter/10.1007/978-3-031-72315-5_12",
  doi       = "10.1007/978-3-031-72315-5\_12"
}

@INPROCEEDINGS{Long2024LLMsynth,
  title     = "{On LLMs-driven synthetic data generation, curation, and
               evaluation: A survey}",
  author    = "Long, Lin and Wang, Rui and Xiao, Ruixuan and Zhao, Junbo and
               Ding, Xiao and Chen, Gang and Wang, Haobo",
  booktitle = "{Findings of the Association for Computational Linguistics ACL
               2024}",
  publisher = "Association for Computational Linguistics",
  pages     = "11065--11082",
  year      =  2024,
  optionalurl       = "http://dx.doi.org/10.18653/v1/2024.findings-acl.658",
  doi       = "10.18653/v1/2024.findings-acl.658"
}

@INPROCEEDINGS{Hod2025surrogate,
  title     = "{Do You Really Need Public Data? Surrogate Public Data for
               Differential Privacy on Tabular Data}",
  author    = "Hod, Shlomi and Rosenblatt, Lucas and Stoyanovich, Julia",
  booktitle = "{The Thirty-ninth Annual Conference on Neural Information
               Processing Systems Datasets and Benchmarks Track}",
  year      =  2025,
  optionalurl       = "https://openreview.net/forum?id=BpT5w97FsP"
}

@BOOK{Griffiths2024SDC,
  title     = "{Handbook on Statistical Disclosure Control for Outputs}",
  author    = "Griffiths, Emily and Greci, Carlotta and Kotrotsios, Yannis and
               Parker, Simon and Scott, James and Welpton, Richard and Wolters,
               Arne and Woods, Christine",
  publisher = "Safe Data Access Professionals Working Group",
  year      =  2024
}

@INCOLLECTION{Green2024statbarn,
  title     = "{The statbarn: A new model for output statistical disclosure
               control}",
  author    = "Green, Elizabeth and Ritche, Felix and White, Paul",
  booktitle = "{Lecture Notes in Computer Science}",
  publisher = "Springer Nature Switzerland",
  pages     = "284--293",
  series    = "Lecture notes in computer science",
  year      =  2024,
  optionalurl       = "http://dx.doi.org/10.1007/978-3-031-69651-0_19",
  doi       = "10.1007/978-3-031-69651-0\_19"
}

@INPROCEEDINGS{Cretu2022QuerySnout,
  title     = "{QuerySnout: Automating the discovery of attribute inference
               attacks against query-based systems}",
  author    = "Cretu, Ana-Maria and Houssiau, Florimond and Cully, Antoine and
               de Montjoye, Yves-Alexandre",
  booktitle = "{Proceedings of the 2022 ACM SIGSAC Conference on Computer and
               Communications Security}",
  publisher = "ACM",
  year      =  2022,
  optionalurl       = "https://dl.acm.org/doi/10.1145/3548606.3560581",
  doi       = "10.1145/3548606.3560581"
}

@INCOLLECTION{Dwork2006approxDP,
  title     = "{Our data, ourselves: Privacy via distributed noise generation}",
  author    = "Dwork, Cynthia and Kenthapadi, Krishnaram and McSherry, Frank and
               Mironov, Ilya and Naor, Moni",
  booktitle = "{Advances in Cryptology - EUROCRYPT 2006}",
  publisher = "Springer Berlin Heidelberg",
  pages     = "486--503",
  series    = "Lecture Notes in Computer Science",
  year      =  2006,
  optionalurl       = "https://dl.acm.org/doi/10.1007/11761679_29",
  doi       = "10.1007/11761679\_29"
}

@MISC{Gomez2025GDP,
  title         = "{Gaussian DP for reporting differential privacy guarantees in
                   machine learning}",
  author        = "Gomez, Juan Felipe and Kulynych, Bogdan and Kaissis, Georgios
                   and Calmon, Flavio P and Hayes, Jamie and Balle, Borja and
                   Honkela, Antti",
  year          =  2025,
  optionalurl           = "http://arxiv.org/abs/2503.10945",
  archivePrefix = "arXiv",
  primaryClass  = "cs.LG",
  eprint        = "2503.10945"
}

@INCOLLECTION{Ogata2017BR,
  title     = "{Learning analytics for E-book-based educational big data in
               higher education}",
  author    = "Ogata, Hiroaki and Oi, Misato and Mohri, Kousuke and Okubo,
               Fumiya and Shimada, Atsushi and Yamada, Masanori and Wang,
               Jingyun and Hirokawa, Sachio",
  booktitle = "{Smart Sensors at the IoT Frontier}",
  publisher = "Springer International Publishing",
  pages     = "327--350",
  year      =  2017,
  optionalurl       = "http://dx.doi.org/10.1007/978-3-319-55345-0_13",
  doi       = "10.1007/978-3-319-55345-0\_13"
}

@INPROCEEDINGS{Hsu2023chronotypes,
  title     = "{Chronotypes of Learning Habits in Weekly Math Learning of Junior
               High School}",
  author    = "Hsu, Chia-Yu and Otgonbaatar, Mandukhai and Horikoshi, Izumi and
               Li, Huiyong and Majumdar, Rwitajit and Ogata, Hiroaki",
  booktitle = "{Proceedings of the 31st International Conference on Computers in
               Education}",
  publisher = "Asia-Pacific Society for Computers in Education",
  volume    =  1,
  pages     = "566--568",
  year      =  2023,
  optionalurl       = "https://library.apsce.net/index.php/ICCE/article/view/4717"
}

@INPROCEEDINGS{Kingma2014SSL,
  title     = "{Semi-supervised Learning with Deep Generative Models}",
  author    = "Kingma, Diederik P and Rezende, Danilo J and Mohamed, Shakir and
               Welling, Max",
  editor    = "Ghahramani, Z and Welling, M and Cortes, C and Lawrence, N and
               Weinberger, K Q",
  booktitle = "{Advances in Neural Information Processing Systems}",
  publisher = "Curran Associates, Inc.",
  volume    =  27,
  year      =  2014,
  optionalurl       = "https://proceedings.neurips.cc/paper_files/paper/2014/file/6d42b1217a6996997ead5a8398c1f944-Paper.pdf"
}

@INPROCEEDINGS{Yousefpour2021Opacus,
  title     = "{Opacus: User-Friendly Differential Privacy Library in {PyTorch}}",
  author    = "Yousefpour, Ashkan and Shilov, Igor and Sablayrolles, Alexandre
               and Testuggine, Davide and Prasad, Karthik and Malek, Mani and
               Nguyen, John and Ghosh, Sayan and Bharadwaj, Akash and Zhao,
               Jessica and Cormode, Graham and Mironov, Ilya",
  booktitle = "{NeurIPS 2021 Workshop Privacy in Machine Learning}",
  year      =  2021,
  optionalurl       = "https://openreview.net/forum?id=EopKEYBoI-"
}

@INPROCEEDINGS{Abadi2016DPSGD,
  title     = "{Deep Learning with Differential Privacy}",
  author    = "Abadi, Martin and Chu, Andy and Goodfellow, Ian and McMahan, H
               Brendan and Mironov, Ilya and Talwar, Kunal and Zhang, Li",
  booktitle = "{Proceedings of the 2016 ACM SIGSAC Conference on Computer and
               Communications Security}",
  publisher = "ACM",
  pages     = "308--318",
  year      =  2016,
  optionalurl       = "https://dl.acm.org/doi/10.1145/2976749.2978318",
  doi       = "10.1145/2976749.2978318"
}

@INPROCEEDINGS{Kingma2015Adam,
  title     = "{Adam: A method for stochastic optimization}",
  author    = "Kingma, Diederik P and Ba, Jimmy",
  booktitle = "{Proceedings of the 3rd International Conference on Learning
               Representations (ICLR 2015)}",
  year      =  2015,
  optionalurl       = "https://arxiv.org/abs/1412.6980"
}

@INPROCEEDINGS{Mironov2017RDP,
  title     = "{R\'{e}nyi Differential Privacy}",
  author    = "Mironov, Ilya",
  booktitle = "{2017 IEEE 30th Computer Security Foundations Symposium (CSF)}",
  publisher = "IEEE",
  pages     = "263--275",
  year      =  2017,
  optionalurl       = "https://ieeexplore.ieee.org/abstract/document/8049725",
  doi       = "10.1109/csf.2017.11"
}

@ARTICLE{Stenger2024eval,
  title     = "{Evaluation is key: a survey on evaluation measures for synthetic
               time series}",
  author    = "Stenger, Michael and Leppich, Robert and Foster, Ian and Kounev,
               Samuel and Bauer, Andr\'{e}",
  journal   = "Journal of big data",
  publisher = "Springer Science and Business Media LLC",
  volume    =  11,
  number    =  1,
  pages     = "1--56",
  year      =  2024,
  optionalurl       = "http://dx.doi.org/10.1186/s40537-024-00924-7",
  doi       = "10.1186/s40537-024-00924-7"
}

@ARTICLE{Rosenblatt2023parity,
  title     = "{Epistemic parity: Reproducibility as an evaluation metric for
               differential privacy}",
  author    = "Rosenblatt, Lucas and Herman, Bernease and Holovenko, Anastasia
               and Lee, Wonkwon and Loftus, Joshua and McKinnie, Elizabeth and
               Rumezhak, Taras and Stadnik, Andrii and Howe, Bill and
               Stoyanovich, Julia",
  journal   = "Proceedings of the VLDB Endowment International Conference on
               Very Large Data Bases",
  publisher = "Association for Computing Machinery (ACM)",
  volume    =  16,
  number    =  11,
  pages     = "3178--3191",
  year      =  2023,
  optionalurl       = "https://dl.acm.org/doi/10.14778/3611479.3611517",
  doi       = "10.14778/3611479.3611517"
}

@INPROCEEDINGS{Salem2023privacygames,
  title     = "{SoK: Let the privacy games begin! A unified treatment of data
               inference privacy in machine learning}",
  author    = "Salem, Ahmed and Cherubin, Giovanni and Evans, David and
               K{\"{o}}pf, Boris and Paverd, Andrew and Suri, Anshuman and
               Tople, Shruti and Zanella-B\'{e}guelin, Santiago",
  booktitle = "{2023 IEEE Symposium on Security and Privacy (SP)}",
  publisher = "IEEE",
  pages     = "327--345",
  year      =  2023,
  optionalurl       = "http://dx.doi.org/10.1109/SP46215.2023.10179281",
  doi       = "10.1109/sp46215.2023.10179281"
}

@INPROCEEDINGS{Yeom2018privacyrisk,
  title     = "{Privacy risk in machine learning: Analyzing the connection to
               overfitting}",
  author    = "Yeom, Samuel and Giacomelli, Irene and Fredrikson, Matt and Jha,
               Somesh",
  booktitle = "{2018 IEEE 31st Computer Security Foundations Symposium (CSF)}",
  publisher = "IEEE",
  pages     = "268--282",
  year      =  2018,
  optionalurl       = "https://www.computer.org/csdl/proceedings-article/csf/2018/668001a268/12OmNyQGSca",
  doi       = "10.1109/csf.2018.00027"
}

@INPROCEEDINGS{Carlini2022LiRA,
  title     = "{Membership Inference Attacks From First Principles}",
  author    = "Carlini, Nicholas and Chien, Steve and Nasr, Milad and Song,
               Shuang and Terzis, Andreas and Tram\`{e}r, Florian",
  booktitle = "{2022 IEEE Symposium on Security and Privacy (SP)}",
  publisher = "IEEE",
  pages     = "1897--1914",
  year      =  2022,
  optionalurl       = "http://dx.doi.org/10.1109/SP46214.2022.9833649",
  doi       = "10.1109/SP46214.2022.9833649"
}

@ARTICLE{Neyman1933NP,
  title     = "{IX. On the problem of the most efficient tests of statistical
               hypotheses}",
  author    = "Neyman, Jerzy and Pearson, Egon S",
  journal   = "Philosophical transactions of the Royal Society of London",
  publisher = "The Royal Society",
  volume    =  231,
  number    = "694-706",
  pages     = "289--337",
  year      =  1933,
  optionalurl       = "https://royalsocietypublishing.org/doi/10.1098/rsta.1933.0009",
  doi       = "10.1098/rsta.1933.0009"
}

@INPROCEEDINGS{Kaissis2024regret,
  title     = "{Beyond the calibration point: Mechanism comparison in
               differential privacy}",
  author    = "Kaissis, G and Kolek, Stefan and Balle, Borja and Hayes, Jamie
               and Rueckert, D",
  editor    = "Salakhutdinov, Ruslan and Kolter, Zico and Heller, Katherine and
               Weller, Adrian and Oliver, Nuria and Scarlett, Jonathan and
               Berkenkamp, Felix",
  booktitle = "{Proceedings of the 41st International Conference on Machine
               Learning}",
  publisher = "PMLR",
  volume    =  235,
  pages     = "22840--22860",
  year      =  2024,
  optionalurl       = "https://proceedings.mlr.press/v235/kaissis24a.html"
}

@INPROCEEDINGS{Bergstra2011TPE,
  title     = "{Algorithms for Hyper-Parameter Optimization}",
  author    = "Bergstra, James and Bardenet, R\'{e}mi and Bengio, Yoshua and
               K\'{e}gl, Bal\'{a}zs",
  editor    = "Shawe-Taylor, J and Zemel, R and Bartlett, P and Pereira, F and
               Weinberger, K Q",
  booktitle = "{Advances in Neural Information Processing Systems}",
  publisher = "Curran Associates, Inc.",
  volume    =  24,
  year      =  2011,
  optionalurl       = "https://proceedings.neurips.cc/paper_files/paper/2011/file/86e8f7ab32cfd12577bc2619bc635690-Paper.pdf"
}

@INPROCEEDINGS{Akiba2019Optuna,
  title     = "{Optuna: A Next-generation Hyperparameter Optimization Framework}",
  author    = "Akiba, Takuya and Sano, Shotaro and Yanase, Toshihiko and Ohta,
               Takeru and Koyama, Masanori",
  booktitle = "{Proceedings of the 25th ACM SIGKDD International Conference on
               Knowledge Discovery \& Data Mining}",
  publisher = "ACM",
  pages     = "2623--2631",
  year      =  2019,
  optionalurl       = "http://dx.doi.org/10.1145/3292500.3330701",
  doi       = "10.1145/3292500.3330701"
}

@INPROCEEDINGS{Balle2018Gaussian,
  title     = "{Improving the Gaussian Mechanism for Differential Privacy:
               Analytical Calibration and Optimal Denoising}",
  author    = "Balle, Borja and Wang, Yu-Xiang",
  booktitle = "{International Conference on Machine Learning}",
  publisher = "PMLR",
  pages     = "394--403",
  year      =  2018,
  optionalurl       = "https://proceedings.mlr.press/v80/balle18a.html"
}

@ARTICLE{Shimizu2011DirectLiNGAM,
  title     = "{DirectLiNGAM: A direct method for learning a linear non-Gaussian
               structural equation model}",
  author    = "Shimizu, Shohei and Inazumi, Takanori and Sogawa, Yasuhiro and
               Hyv{\"{a}}rinen, Aapo and Kawahara, Yoshinobu and Washio, Takashi
               and Hoyer, Patrik O and Bollen, Kenneth",
  journal   = "Journal of Machine Learning Research",
  publisher = "JMLR.orgPUB6573",
  volume    =  12,
  pages     = "1225--1248",
  year      =  2011,
  optionalurl       = "https://dl.acm.org/doi/10.5555/1953048.2021040",
  doi       = "10.5555/1953048.2021040"
}

@ARTICLE{Leyder2024TSLiNGAM,
  title     = "{TSLiNGAM: DirectLiNGAM under heavy tails}",
  author    = "Leyder, Sarah and Raymaekers, Jakob and Verdonck, Tim",
  journal   = "Journal of Computational and Graphical Statistics: A Joint
               Publication of American Statistical Association, Institute of
               Mathematical Statistics, Interface Foundation of North America",
  publisher = "Informa UK Limited",
  pages     = "1--20",
  year      =  2024,
  optionalurl       = "http://dx.doi.org/10.1080/10618600.2024.2394462",
  doi       = "10.1080/10618600.2024.2394462"
}

@INPROCEEDINGS{Dibia2025DPusability,
  title     = "{SoK: Usability studies in differential privacy}",
  author    = "Dibia, Onyinye and Bhattacharjee, Prianka and Stenger, Brad and
               Baldasty, Steven and Bates, Mako and Ngong, Ivoline and Feng,
               Yuanyuan and Near, Joseph",
  booktitle = "{Proceedings on Privacy Enhancing Technologies Symposium}",
  publisher = "Privacy Enhancing Technologies Symposium Advisory Board",
  volume    =  2025,
  pages     = "881--895",
  series    =  4,
  year      =  2025,
  optionalurl       = "https://petsymposium.org/popets/2025/popets-2025-0162.php",
  doi       = "10.56553/popets-2025-0162"
}

@INPROCEEDINGS{Cyffers2025DPsetting,
  title     = "{Setting $\varepsilon$ is not the Issue in Differential Privacy}",
  author    = "Cyffers, Edwige",
  booktitle = "{The Thirty-Ninth Annual Conference on Neural Information
               Processing Systems Position Paper Track}",
  year      =  2025,
  optionalurl       = "https://openreview.net/forum?id=tp94g4Vmad"
}

@TECHREPORT{Hubbard2020TRE,
  title       = "{Trusted Research Environments (TRE) Green Paper}",
  author      = "Hubbard, Tim and Reilly, Gerry and Varma, Susheel and Seymour,
                 David",
  publisher   = "Zenodo",
  institution = "UK Health Data Research Alliance",
  year        =  2020,
  optionalurl         = "https://zenodo.org/records/4594704",
  doi         = "10.5281/ZENODO.4594704"
}

@ARTICLE{Nab2024OpenSAFELY,
  title     = "{OpenSAFELY: A platform for analysing electronic health records
               designed for reproducible research}",
  author    = "Nab, Linda and Schaffer, Andrea L and Hulme, William and DeVito,
               Nicholas J and Dillingham, Iain and Wiedemann, Milan and Andrews,
               Colm D and Curtis, Helen and Fisher, Louis and Green, Amelia and
               Massey, Jon and Walters, Caroline E and Higgins, Rose and
               Cunningham, Christine and Morley, Jessica and Mehrkar, Amir and
               Hart, Liam and Davy, Simon and Evans, David and Hickman, George
               and Inglesby, Peter and Morton, Caroline E and Smith, Rebecca M
               and Ward, Tom and O'Dwyer, Thomas and Maude, Steven and Bridges,
               Lucy and Butler-Cole, Ben F C and Stables, Catherine L and
               Stokes, Pete and Bates, Chris and Cockburn, Jonny and Hester,
               Frank and Parry, John and Bhaskaran, Krishnan and Schultze, Anna
               and Rentsch, Christopher T and Mathur, Rohini and Tomlinson,
               Laurie A and Williamson, Elizabeth J and Smeeth, Liam and Walker,
               Alex and Bacon, Sebastian and MacKenna, Brian and Goldacre, Ben",
  journal   = "Pharmacoepidemiology and Drug Safety",
  publisher = "Wiley",
  volume    =  33,
  number    =  6,
  pages     = "e5815",
  year      =  2024,
  optionalurl       = "http://dx.doi.org/10.1002/pds.5815",
  doi       = "10.1002/pds.5815"
}

@INPROCEEDINGS{Shin2017replay,
  title     = "{Continual Learning with Deep Generative Replay}",
  author    = "Shin, Hanul and Lee, Jung Kwon and Kim, Jaehong and Kim, Jiwon",
  editor    = "Guyon, I and Luxburg, U Von and Bengio, S and Wallach, H and
               Fergus, R and Vishwanathan, S and Garnett, R",
  booktitle = "{Advances in Neural Information Processing Systems}",
  publisher = "Curran Associates, Inc.",
  volume    =  30,
  year      =  2017,
  optionalurl       = "https://proceedings.neurips.cc/paper_files/paper/2017/file/0efbe98067c6c73dba1250d2beaa81f9-Paper.pdf"
}

@INPROCEEDINGS{vanDenOord2017VQ,
  title     = "{Neural Discrete Representation Learning}",
  author    = "van den Oord, Aaron and Vinyals, Oriol and Kavukcuoglu, Koray",
  editor    = "Guyon, I and Luxburg, U Von and Bengio, S and Wallach, H and
               Fergus, R and Vishwanathan, S and Garnett, R",
  booktitle = "{Advances in Neural Information Processing Systems}",
  publisher = "Curran Associates, Inc.",
  volume    =  30,
  year      =  2017,
  optionalurl       = "https://proceedings.neurips.cc/paper_files/paper/2017/file/7a98af17e63a0ac09ce2e96d03992fbc-Paper.pdf"
}

@ARTICLE{Srivastava2014Dropout,
  title   = "{Dropout: a simple way to prevent neural networks from overfitting}",
  author  = "Srivastava, Nitish and Hinton, Geoffrey E and Krizhevsky, A and
             Sutskever, I and Salakhutdinov, R",
  journal = "Journal of Machine Learning Research",
  volume  =  15,
  number  =  56,
  pages   = "1929--1958",
  year    =  2014,
  optionalurl     = "http://dx.doi.org/10.5555/2627435.2670313",
  doi     = "10.5555/2627435.2670313"
}

@INPROCEEDINGS{Fu2019annealing,
  title     = "{Cyclical Annealing Schedule: A Simple Approach to Mitigating KL
               Vanishing}",
  author    = "Fu, Hao and Li, Chunyuan and Liu, Xiaodong and Gao, Jianfeng and
               Celikyilmaz, Asli and Carin, Lawrence",
  editor    = "Burstein, Jill and Doran, Christy and Solorio, Thamar",
  booktitle = "{Proceedings of the 2019 Conference of the North American Chapter
               of the Association for Computational Linguistics: Human Language
               Technologies, Volume 1 (Long and Short Papers)}",
  publisher = "Association for Computational Linguistics",
  pages     = "240--250",
  year      =  2019,
  optionalurl       = "http://dx.doi.org/10.18653/v1/N19-1021",
  doi       = "10.18653/v1/N19-1021"
}

\newpage

 





\end{document}